\documentclass[10pt, conference, letterpaper]{IEEEtran}
\usepackage{cite}

\ifCLASSINFOpdf
\else
\fi
\usepackage{algorithm}
\usepackage[noend]{algpseudocode}
\usepackage{amsmath}
\usepackage{amssymb}
\usepackage{amsthm}
\usepackage{xcolor}
\usepackage{graphicx}
\usepackage{physics}
\usepackage{enumitem}

\newtheorem{theorem}{Theorem}
\theoremstyle{remark}
\newtheorem{remark}{Remark}

\newtheorem{proposition}{Proposition}

\begin{document}
%
\title{Utility-Based Path Selection and Configuration in Quantum Networks via Layered Shortest Paths}

\author{\IEEEauthorblockN{Leonardo Bacciottini, Subhransu Maji, Don Towsley, Gayane Vardoyan}
\IEEEauthorblockA{Manning College of Information and Computer Sciences\\
University of Massachusetts Amherst\\
Email: \{lbacciottini,smaji,dtowsley,gvardoyan\}@umass.edu}
}

%





\maketitle

\begin{abstract}A path in a quantum network is a chain of repeaters that distributes entanglement between two users. Selecting a path requires balancing the rate and quality (e.g., fidelity) of the delivered entanglement, but these quantities, unlike standard routing metrics, compose non-additively. The problem is compounded by link-level configuration choices (e.g., distillation rounds or emitter brightness tuning), each trading rate against fidelity, so that a path's performance depends jointly on its route and its per-link settings.
We cast this joint path selection and configuration problem as a shortest-path computation on a layered graph whose layers track discretized end-to-end fidelity. A single run returns the full rate–fidelity Pareto frontier, from which the path maximizing any nondecreasing utility function of rate and fidelity can be selected. We prove that for certain utility functions (including the secret key rate of BB84), the method is a fully polynomial-time approximation scheme, returning a near-optimal path within a specified tolerance. We further characterize exactly when cheaper scalarization-based routing suffices: it is optimal for utility functions with convex fidelity profiles, but can be arbitrarily suboptimal otherwise (e.g., for step-like, sigmoidal utilities), whereas the layered method remains reliable in all cases.
\end{abstract}



%
\IEEEpeerreviewmaketitle

\section{Introduction}
Quantum networks distribute entanglement to support
applications such as provably secure communication~\cite{2014.TheoCompSc.Bennett-Brassard.BB84,1991.PRL.Ekert.QCrypt}, quantum-enhanced sensing~\cite{2012.PRL.Gottesman-Croke.QRepTelescope,khabiboulline2019quantum,2014.NatPhys.Komar-Lukin.QuantClocks,giovannetti2011advances}, verifiable blind quantum computation~\cite{fitzsimons2017unconditionally}, and distributed quantum computation~\cite{jiang2007distributed} with up to exponential speedups compared to its conventional counterpart.
Quantum network performance depends critically on how entanglement is routed through chains of lossy and noisy quantum channels.
Unlike
standard classical routing metrics, however, the two principal measures
of an entanglement path---the rate and quality (e.g., fidelity) of the delivered
end-to-end states---compose non-additively and generally trade off
against one another.

Optimizing either quantity alone can therefore lead to a poor
application-level solution. A high-rate path may produce states whose
fidelity is too low to be useful, while a high-fidelity path may
distribute entanglement too slowly. Application objectives such as the
secret key rate of a quantum key distribution (QKD) protocol depend jointly on both quantities and are better
described by a \emph{quantum utility function}~\cite{vardoyan2023quantum}. Selecting a
path according to such a utility allows the routing decision to reflect
the requirements of the application being executed, rather than optimizing a
single intermediate performance measure such as rate, fidelity, or hop count.

The routing decision is further complicated by configurable
rate--fidelity tradeoffs at individual links. For example, the
bright-state population in single-click entanglement generation protocols \cite{cabrillo1999creation}, or the number of
entanglement distillation rounds can be adjusted to trade generation
rate for state quality (see more discussion in \S~\ref{sec:model}). We refer to the available configuration choices
on a link as its \emph{operating points}. Consequently, the performance
of an end-to-end connection depends jointly on the selected physical
route and the operating point chosen on every traversed link. 

We study this joint path-selection and link-configuration problem for
bipartite entanglement distribution,
where the target is an ideal Bell state to be shared between two nodes in the network. Given a pair of nodes \(s\) and
\(t\), and a nondecreasing utility function \(U(R,F)\) of end-to-end
rate $R$ and fidelity $F$, our goal is to select an \(s\)--\(t\) path and one
operating point on each link on the path to maximize the resulting
utility. Our formulation builds on the quantum network utility
maximization framework~\cite{vardoyan2023quantum}, but focuses on path selection
and explicitly incorporates discrete per-link configuration choices.

For clarity of presentation, we develop our formulation under the assumption of Werner states~\cite{werner1989quantum} at the elementary link level. 
Werner states are commonly used to model worst-case Pauli noise in quantum networks~\cite{nielsen2010quantum}, thus enabling studies of performance lower-bounds assuming Pauli channels. 
A Werner state can be described in terms of its parameter $W$, which is linearly related to its fidelity $F$ (see \S~\ref{sec:model}).

\if{false}
Our key observation is that the adopted quantum path-composition rules
become additive after a suitable reparameterization. Under the
one-shot entanglement distribution model \lb{This is unclear at this point}, inverse rates add along a path, while Werner
parameters multiply under entanglement swapping (see \S~\ref{sec:model} for details). Defining
$(X,Y)=(1/R, -\ln W)$, where $W$ \lb{I would avoid going in such a detail in the intro. Additionally, I don't think we should commit to inverse rates here} is the Werner parameter of a state generated at the link level, therefore converts both quantities into additive path metrics. We exploit this structure by constructing a layered graph whose layers
discretize the cumulative Werner-parameter coordinate \(Y\), while its
edge costs accumulate inverse rate \(X\) \lb{how about simply "its edge costs track rate", so we are agnostic at this level w.r.t. the rate composition rule}. Running Dijkstra's algorithm~\cite{dljkstra1959note}
on this graph computes, for every node and fidelity layer, the path
with the smallest inverse rate \lb{"the highest-rate path"}. At the destination, these paths form a
discretized rate--fidelity frontier that can be evaluated under the
desired utility function.
\fi

Our key observation is that both rate and fidelity become additive properties after a suitable reparameterization. We exploit this structure though a layered graph whose layers discretize the end-to-end fidelity, while its
edge costs track the corresponding rate. Running Dijkstra's algorithm~\cite{dljkstra1959note}
on this graph computes the highest-rate path for every node and fidelity layer.
This construction yields the \emph{quantum layered shortest-path}
(QLSP) algorithm. In contrast to methods that collapse rate and
fidelity into a single edge cost, QLSP retains the end-to-end fidelity
of the state explicitly and can therefore recover Pareto-optimal solutions
that are not supported by any weighted-sum scalarization of functions of rate and fidelity. Moreover, given a node $s$, a
single run computes the discretized frontier for \emph{every} network node $t$ and
can be reused for \textit{any} nondecreasing utility function. We
analyze the approximation introduced by fidelity discretization and
provide polynomial-time approximation guarantees under certain assumptions (see Theorem~\ref{theorem:approximation} in \S~\ref{sec:methods}). 

\if{false}
We also characterize when the cheaper weighted-sum approach is
sufficient. For utilities of the form $U(R,W)=R\,g(W)$, 

define the fidelity profile $\varphi(y)=g(e^{-y})$. We prove that if \(\varphi\) is convex, then some utility-optimal solution is supported: it minimizes $\alpha X+\beta Y$
for an appropriate pair of nonnegative weights $(\alpha,\beta)\geq0$ (see Theorem~\ref{thm:main} in \S~\ref{sec:methods}). Consequently, a
weighted-sum procedure that enumerates all supported vertices recovers
the exact optimum.
\fi

We also establish a sufficient condition under which a weighted-sum scalarization recovers the utility-optimal route and link configuration (Theorem~\ref{thm:main}, \S~\ref{sec:methods}). This result motivates a computationally cheaper baseline that is provably optimal whenever this condition holds.
This condition holds for 
the asymptotic secret key rate of the BB84 QKD protocol \cite{2014.TheoCompSc.Bennett-Brassard.BB84}, as well as a utility based on entanglement negativity~\cite{vidal2002computable}. However, nonconvex fidelity profiles,
such as sigmoid functions or those induced by hard fidelity requirements, can have
Pareto-optimal but unsupported maximizers that no weighted-sum
scalarization can recover. QLSP continues to retain such solutions
because it approximates the full rate--fidelity frontier rather than
only its supported convex-hull vertices. 

Our main contributions are as follows:
\begin{itemize}[noitemsep,topsep=0pt,leftmargin=*]

    \item We introduce QLSP, a layered shortest-path algorithm that solves the joint path selection and configuration problem
    by computing the rate--fidelity Pareto frontier from a given node $s$ to every node $t$ in a single
    shortest-path computation.

    \item We analyze the complexity and approximation
    quality of QLSP, and establish polynomial-time approximation
    guarantees under the assumptions stated in
    \S~\ref{sec:methods}.

    \item We characterize when a linear scalarization is exact, and propose
    a cheaper algorithm that exploits this property.
    
\end{itemize}
 We evaluate QLSP on both synthetic and real-world network
    topologies, including heterogeneous architectures. 
    The experiments show that in the considered cases there are Pareto-optimal 
    candidate solutions that no linear scalarization can find. 
    We also observe that QLSP outperforms all other considered baselines, including those that optimize only the end-to-end fidelity or rate.
    
\if{false}\begin{itemize}[noitemsep,topsep=0pt,leftmargin=*]
    \item We introduce an efficient approximation algorithm to solve the joint path selection and link configuration problem. Our algorithm is based on Dijkstra's algorithm deployed on a layered graph of a given quantum network with arbitrary topology. The layered graph is constructed to capture a sufficiently large number 
    quantum-utility eventualities. We call this the \emph{quantum layered shortest path} (QLSP) algorithm, and show that its runtime is polynomial in the size of the quantum network, the number of operating points, and the quality of the approximation; 
    \item We provide a detailed analysis of the QLSP algorithm's approximation quality;
    \item As an independent contribution, we show that for utility functions that have convex fidelity profiles, deploying Dijkstra's algorithm on the original quantum network graph is optimal for appropriate choices of weights combining said performance measures used as edge costs.
\end{itemize}
\fi

\section{Related Work}
\label{sec:related}

We organize prior work based on treatment of the rate--fidelity tradeoff. A comprehensive taxonomy of entanglement routing appears in the recent survey of Abane et al.~\cite{abane2025entanglement}.

\paragraph{Rate- and fidelity-aware quantum routing}
Early work on repeater-network path selection ranks paths using additive per-link costs derived from simulated link throughputs~\cite{2013_vanMeter_path_selection}, an approach representative of the edge-based baselines considered in \S~\ref{sec:methods}. Subsequent methods optimize end-to-end rate or network-wide throughput, including single-path, multipath, and concurrent-flow formulations~\cite{2012_DiFranco_OptimalPath,caleffi2017optimal,2016.ArXiv.Schoute-Wehner.PerfQnetroute,pant2019routing,shi2020concurrent,chakraborty2020entanglement,ghaderibaneh2022efficient,vardoyan2024bipartite}. Fidelity is typically ignored or imposed as a hard constraint. For example, Q-PATH/Q-LEAP jointly selects a path and per-link purification rounds subject to an end-to-end fidelity threshold~\cite{a10}, while later work jointly optimizes routing and purification to maximize throughput under a fidelity requirement~\cite{zhao2022e2e,xiao2024purification,victora2023entanglement}. Related scheduling methods allocate resources to satisfy rate and fidelity requirements on fixed paths~\cite{SWarxivQoS}. Such fidelity-constrained objectives are special cases of our framework, represented by utilities such as
$U(R,F)=R\;\mathbf{1}[F\ge F_{\mathrm{req}}]$, where $R$ is the rate of entanglement generation, $F$ is the fidelity to the desired entangled state, and $F_{\mathrm{req}}$ is a pre-specified fidelity threshold imposed on each distributed state.

\paragraph{Utility-based routing}
Our objective builds on the QNUM framework of Vardoyan et al.~\cite{vardoyan2023quantum}, which adapts classical network utility maximization~\cite{NUM1,NUM2,NUM3} to entanglement distribution over fixed routes, using either centralized~\cite{vardoyan2023quantum} or distributed~\cite{panigrahy2025framework} optimization. Most relevant, Kar and Mukhopadhyay~\cite{kar2026utility} jointly optimize routes and continuous rate--fidelity allocations through a mixed-integer convex formulation. They obtain exact solutions for negativity-based utilities in a high-rate regime, and otherwise use a convex relaxation. 
Using a discrete set of rate--fidelity configurations, Zhang et al.~\cite{zhang2025linkconfig} jointly optimize path selection, per-link configuration, and purification for \emph{multiple} rate demands under a hard end-to-end fidelity constraint, solved with a greedy shortest-path heuristic refined by Bayesian optimization. Our setting considers a single pair of end nodes wishing to share entanglement, but allows arbitrary monotone utility functions and comes with approximation guarantees unlike theirs.

\paragraph{Multi-criteria shortest paths in classical networking}
Routing under two additive metrics is a classical problem in graph theory. In particular, finding a minimum-cost path subject to a delay constraint is NP-hard~\cite{wang1996quality}, admits fully polynomial-time approximation schemes (FPTASs) based on rounding-and-scaling dynamic programs~\cite{hassin1992approximation,lorenz2001simple}, and is often addressed in practice using Lagrangian relaxation, most notably the LARAC algorithm~\cite{juttner2001lagrange}. LARAC efficiently searches over the multiplier space, but it inherits the duality gap associated with scalarization; moreover, the number of supported solutions it may need to explore can be superpolynomial~\cite{carstensen1983complexity}.
This also renders techniques that maintain the pareto-frontier explicitly more challenging  (e.g., Coutinho et al.~\cite{coutinho2024fidelitycurves}), as the size of the frontier is not bounded \emph{a priori}.
QLSP adapts the rounding-and-scaling framework to quantum networks through the reparameterization of rate and fidelity which makes the quantum composition rules additive: rates compose harmonically (under certain assumptions, see \S~\ref{sec:model}), while fidelities compose multiplicatively. 
It further incorporates discrete per-link operating points, supports arbitrary monotone utilities, and provides an exact characterization of optimality (Theorem~\ref{thm:main}).

\section{Model}
\label{sec:model}
\if{false}
\begin{figure}
    \centering
    \includegraphics[width=0.98\linewidth]{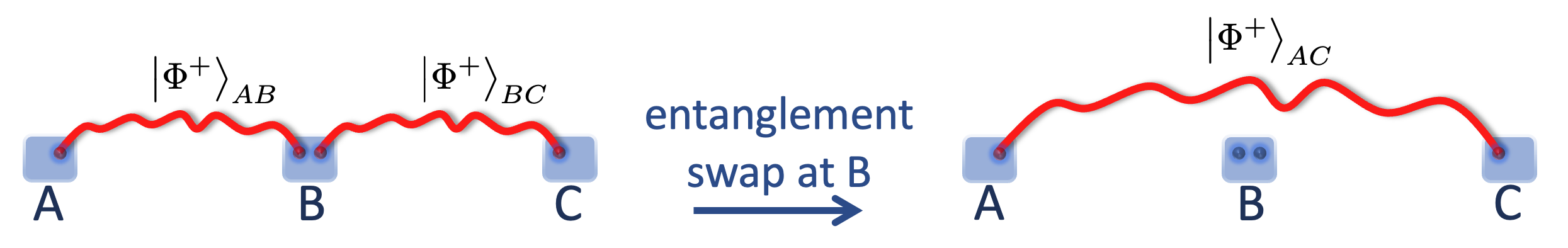}
    \vspace{-1em}
    \caption{To entangle node $A$ and $C$'s qubits, link-level entanglement is first generated between each node and an intermediate node $B$. $B$---acting as a quantum repeater---then carries out a BSM (entanglement swap) on its two locally-held qubits, resulting in $A$-$C$ entanglement.}
    \label{fig:swap}
\end{figure}
\fi
We consider first-generation quantum networks~\cite{MAT2015Qrepeaters} which distribute entanglement to end nodes first by generating link-level entanglement (LLE) between intermediate, physically-adjacent network nodes, called quantum repeaters. Quantum repeaters then ``fuse'' these shorter-distance entangled states into end-to-end (e2e) entanglement via Bell-state measurements (BSMs). 
This process is also known as \emph{entanglement swapping}.
LLE generation (LLEG) is an inherently probabilistic process that degrades exponentially with link length in optical fiber: the success probability $p_{\text{gen}} \propto \exp{-L/L_{\text{att}}}$, where $L$ is link length and $L_{\text{att}}$ is the fiber attenuation length, usually taken to be $22$ km. Unless specified otherwise, we take all link lengths $L$ to be in units of km. If the architecture allows it, distillation protocols~\cite{bennett1996purification, deutsch1996quantum} can also be used at any stage to probabilistically to reduce the number of entangled states and improve quality.

\paragraph*{Quantum noise model}
Realistic quantum systems produce non-ideal quantum states that further degrade due to imperfect quantum storage, manipulation (gates), and measurement. Fidelity is a measure of closeness between quantum states (or gates) and their ideal realizations. We assume in this work that LLEG on link $l$ produces a Werner state
\begin{align}
\rho_l = W_l\ketbra{\Psi^-}+\frac{1-W_l}{4}I_4
\label{eq:werner}
\end{align}
with probability $p_l$, where $W_l\in[0,1]$ is the state's Werner parameter, $\ket{\Psi^-} = (\ket{01}-\ket{10})/\sqrt{2}$ is the desired Bell state, and $I_4$ is the $4\times 4$ identity matrix. The fidelity of $\rho_l$ to $\ket{\Psi^-}$ is then $F_l = (3W_l+1)/4$, and for the state to be considered entangled it is necessary that ${F_l >0.5}$ (equivalently, ${W_l>1/3}$). While in reality LLEG can produce states that are not Werner, we assume their twirling~\cite{Bennett96} into the form (\ref{eq:werner}) to simplify analysis. Namely, Werner states have the convenient property that entanglement swapping on a set of $n$ Werner states $\rho_l$ yields another Werner state with parameter $W=\prod_{l=1}^nW_l$.

\paragraph*{Operating points}
Physical mechanisms for entanglement generation commonly allow generation rate
to be traded for state quality.  This occurs, for example, when tuning the
bright-state population in single-click LLEG~\cite{cabrillo1999creation}, the pump power in
spontaneous parametric down-conversion (SPDC)~\cite{kwiat1999ultrabright}, or
the number of rounds in an entanglement distillation protocol.  Each available
configuration therefore gives a different expected LLEG rate and Werner
parameter.  We call one such rate--quality pair $(R,W)$ an \emph{operating
point}.  

For single-click entanglement generation, the
bright-state population $\alpha\in(0,1)$ controls the rate--fidelity tradeoff.  In the regime of
small link transmissivity $\eta$, the fidelity and success probability are
\begin{equation}
  F(\alpha)=1-\alpha,
  \qquad
  p_{\mathrm{gen}}(\alpha)=2\eta(L)\alpha.\label{eq:single-click-tradeoff}
\end{equation}
Assuming a heralding station placed exactly midway between two nodes separated by $L$ km, the ``link'' is defined as the node-to-heralding station segment
of length $L/2$ km, yielding
\begin{equation}
  \eta(L)=c e^{-L/(2L_{\mathrm{att}})},\label{eq:single-click-transmission}
\end{equation}
where $0<c<1$ collects coupling, conversion, detector, and other efficiencies
not due to fiber attenuation.  If LLEG is attempted at repetition rate
$R_{\mathrm{rep}}$, a setting $\alpha$ maps to the QLSP operating point
\begin{equation}
  R(\alpha)=R_{\mathrm{rep}}p_{\mathrm{gen}}(\alpha),
  \qquad
  W(\alpha)=(4F(\alpha)-1)/3.\label{eq:single-click-operating-point}
\end{equation}
Thus, increasing $\alpha$ raises the generation rate while lowering the
Werner parameter.  Although $\alpha$ is continuous, selecting $A$ admissible
settings $\alpha_1,\ldots,\alpha_A$ produces the finite menu
  $\bigl\{\bigl(R(\alpha_a),W(\alpha_a)\bigr)\bigr\}_{a=1}^{A}$
used by QLSP.

Entanglement distillation produces the same algorithmic interface through a
different physical tradeoff: it consumes multiple lower-quality pairs to
produce fewer pairs of higher quality.  Consider the BBPSSW
protocol~\cite{bennett1996purification}, which distills two identical Werner states
into one.  Suppose that LLEG supplies raw states with Werner parameter $W_0$
at expected rate $R_0$, and allow at most $n$ nested distillation rounds.  Let
$(R_k,W_k)$ denote the expected rate and Werner parameter after $k$ rounds, so
that $k=0$ is the undistilled raw pair.  Under ideal local operations and perfect
storage of quantum states, the
success probability for the next BBPSSW round is
\begin{equation}
  P_{\mathrm{succ}}(W_k)=(1+W_k^2)/2,\label{eq:bbpssw-success}
\end{equation}
and the output Werner parameter and expected rate obey~\cite{dur1999quantum}
\begin{equation}
  W_{k+1}=\frac{2W_k(1+2W_k)}{3(1+W_k^2)},
  \qquad
  R_{k+1}=\frac{P_{\mathrm{succ}}(W_k)}{2}R_k.\label{eq:bbpssw-recurrence}
\end{equation}
Choosing between zero and $n$ nested rounds therefore gives the finite menu
  $\bigl\{(R_k,W_k)\bigr\}_{k=0}^{n}$.
Unlike single-click tuning, the number of rounds is already a discrete
protocol choice, so this entire menu is represented directly.

The single-click scheme and BBPSSW distillation thus have different physical
machinery but expose the same QLSP input: a finite collection of feasible
$(R,W)$ pairs.  Other physical or protocol controls can be incorporated in
the same way.  Sampling a continuous tradeoff more finely enriches the set of
available configurations, at the cost of additional computation; this
physical-menu discretization is distinct from the algorithmic approximation
introduced in \S~\ref{sec:methods}.

\paragraph*{Utility functions} The QLSP algorithm can accommodate arbitrary quantum utility functions $U(R,F)$---where $F$ is the fidelity of the average e2e Werner state distributed to two network nodes and $R$ is the e2e generation rate---provided that $U$ is monotonically increasing in its inputs. The utility function captures the complex ways in which users or quantum applications balance rate and fidelity. As an example, the asymptotic secret key rate (SKR) of the BB84 quantum QKD protocol, when carried out with Werner states of fidelity $F$ distributed at rate $R$, reads
\begin{align}
    U_{\text{BB84}}(R,F) = \max\Big\{0,R\left(1-2h\left(\frac{2(1-F)}{3}\right)\right)\Big\},\label{eq:bb84}
\end{align}
where $h(\cdot)$ is the binary entropy function.
Utility functions can combine any function of rate and an entanglement monotone, e.g., negativity, which for Werner states simplifies to
\begin{align}
    U_{\textrm{neg}}(R,F) = \max\Big\{0,R\left(F-\frac{1}{2}\right)\Big\}.
    \label{eq:neg}
\end{align}
Another commonly used one is the fidelity-threshold utility,
\begin{align}
    U_{\textrm{th}}(R,F) = R\mathbf{1}[F \geq F_{\textrm{req}}]
    \label{eq:th}
\end{align}
which maximizes the generation rate subject to a fidelity requirement
$F_{\mathrm{req}}$.
It may also be desirable to apply transformations to these utilities, such as
a logarithm, to promote fairness among coexisting network
flows~\cite{vardoyan2023quantum}. 
Path selection under such nonlinear objectives is even more challenging when network links have their own configurable rate--fidelity tradeoffs, since optimizing each link according to a local objective does not in general produce an e2e optimum. 

These utilities can equivalently be expressed in terms of the Werner-state
parameter $W$ using the relation $F=(3W+1)/4$. In the remainder of the
paper, we also use this parameterization because it simplifies the analysis.

\paragraph*{Problem statement}

Consider a quantum network represented by a graph \({G=(\mathcal{V},\mathcal{E})}\),
where vertices $v\in\mathcal{V}$ represent nodes (e.g., repeaters and end nodes), and undirected edges $e\in\mathcal{E}$ represent physical network links.
Each edge \(e\in \mathcal{E}\) may admit one or more operating points \(a\in\mathcal{A}_e\). Each operating point is associated with a rate-Werner parameter pair $(R_{e,a}, W_{e,a})$, where \(R_{e,a}>0\) is the average rate at which LLEG produces a new entangled pair, and \(W_{e,a}\in[0,1]\) is the corresponding Werner parameter. 

Given nodes \(s\) and \(t\),
producing an entangled state between them requires selecting a routing solution $\pi = (p, a)$, where:
\begin{enumerate}
    \item $p=(e_1, e_2,\dots,e_n)$ denotes a path from $s$ to $t$, represented as a sequence of edges $e_k \in \mathcal{E}$; and 
    \item $a = (a_1, a_2, \dots, a_n)$ denotes the corresponding operating points, i.e.,  $a_k \in \mathcal{A}_{e_k}$ is the operating point on edge $e_k$.
\end{enumerate}

Each routing solution $\pi=(p,a)$ induces an e2e rate
$R_{\text{e2e}}(\pi)$, representing the average number of entangled states
produced between $s$ and $t$ per unit time, and an e2e Werner parameter
$W_{\text{e2e}}(\pi)$, capturing their quality. Under the models adopted in
this paper, they compose as
\begin{equation}
    R_{\text{e2e}}(\pi)
    = \left(\sum_k \frac{1}{R_{e_k,a_k}}\right)^{-1},
    ~~
    W_{\text{e2e}}(\pi)
    = \prod_k W_{e_k,a_k}.
    \label{eq:path_composition}
\end{equation}
$R_{\text{e2e}}$ is the reciprocal of the expected e2e latency and is exact
when entanglement swaps are performed sequentially along the path. A
bottleneck model would instead use
$R_{\text{e2e}}(\pi)=\min_k R_{e_k,a_k}$. We use the first relation throughout,
although our framework extends to any additive rate model. The second relation ($W_{\text{e2e}}$)
follows from Werner-state composition. Depolarizing gate noise can be included
through additional multiplicative factors; memory decoherence can be included
only when represented by a constant depolarizing term.

Given nodes $s$ and $t$, and a utility function $U(R,W)$ of rate $R$ and Werner parameter $W$, our goal is to find a routing solution $\pi=(p,a)$ that jointly selects and configures an $s$--$t$ path to maximize its end-to-end utility:
\begin{equation}
    \pi^{\star} = \arg\max_{\pi} U(R_{\text{e2e}}(\pi), W_{\text{e2e}}(\pi)).
\label{eq:opt_problem}
\end{equation}

We note that it is possible to use time multiplexing or multi-path routing to obtain solutions that use multiple paths or configurations to get a better objective than~\eqref{eq:opt_problem}, but these techniques are out of the scope of this work.

\section{Methods}
\label{sec:methods}

In this section, we formally present the QLSP algorithm for approximately solving \eqref{eq:opt_problem} and analyze its correctness, computational complexity, and approximation guarantee. We then introduce several baselines that we evaluate against QLSP in \S~\ref{sec:numer}. Finally, in Theorem~\ref{thm:main}, we show that one of these baselines---the \emph{surrogate sweep} algorithm---is optimal for quantum utility functions with convex fidelity profiles.

\subsection{Quantum layered shortest-path (QLSP) formulation}
\label{sec:qlsp}
Because quantum-network utility functions may depend on both e2e rate and fidelity, directly optimizing either quantity alone does not generally yield an optimal solution.
A path with the highest e2e rate, for instance, might yield zero utility if its e2e fidelity is too low (e.g., $\leq 1/2$ for $U_{\text{NEG}}$).
Our formulation (Algorithm~\ref{alg:layered-graph}) instead maintains a discretized Pareto frontier over the achievable rate--fidelity tradeoffs. Specifically, for each node and each discretized value of the e2e Werner parameter, we maintain the best e2e rate among paths terminating at that node. This construction can be viewed as introducing multiple layers for each node, with each layer corresponding to a discretized Werner-parameter value. We then show that computing shortest paths in this layered graph recovers the discretized Pareto frontier, which contains an approximately optimal solution, and we quantify the approximation error.

\paragraph{Path composition rule} We assume that, compatibly with \eqref{eq:path_composition}, extending a path by edge $e$ operating at point $a\in\mathcal{A}$ changes the path's effective rate $R$ and Werner parameter $W$ according to:
\begin{equation}
\frac{1}{R'}=\frac{1}{R}+\frac{1}{R_{e,a}},
\qquad
W'=W\,W_{e,a}.
\end{equation}

It is convenient to reparameterize these quantities as: ${X\equiv1/R}$, and ${Y\equiv-\ln W}$. Under this change of variables, the update becomes additive: 
$X' = X + x_{e,a}$, and $Y' = Y + y_{e,a}$,
where $x_{e,a} := 1/R_{e,a}$, and $y_{e,a} := -\ln W_{e,a}$.

\paragraph{Layered graph construction}
Choose a discretization width \(\Delta>0\) and let the Werner parameter bins be indexed by $k=0,1,\dots,K$. Bin \(k\) represents cumulative fidelity-loss $Y$ approximately equal to \(k\Delta\), corresponding to Werner parameter $W_k \approx e^{-k\Delta}$.
Given a Werner-parameter threshold $W_{\min}$, and $K$ can be set as $K = \lceil \ln(1/W_{\min})/\Delta \rceil$.
For example, one may choose $W_{\min}>1/3$ when optimizing negativity,
since states below this threshold have zero negativity. The value
$W_{\min}$ is not intrinsic to the construction; it simply truncates
the layered search space. 
We construct a layered graph \(G^\Delta=(\mathcal{V}^\Delta,\mathcal{E}^\Delta)\) as follows. The layered node set is
\[
\mathcal{V}^\Delta = \{(v,k) : v\in \mathcal{V},\; k\in\{0,\dots,K\}\}.
\]
Thus, each physical node \(v\) is replicated across \(K+1\) fidelity layers (equivalently, Werner-parameter layers, since there is a one-to-one correspondence between them). For each physical edge \(e=(u,v)\in \mathcal{E}\), each operating point \(a\in\mathcal{A}_e\), and each layer \(k\), define
$k' = \left\lceil \frac{k\Delta + y_{e,a}}{\Delta} \right\rceil$.
If \(k'\le K\), then we add a directed layered edge $(u,k)\to (v,k')$
with nonnegative cost
$
c\big((u,k),(v,k')\big)=x_{e,a}=1/R_{e,a}.
$
If the physical graph is undirected, we add the reverse layered edge as well. We show an example of how the layered graph is constructed in Fig.~\ref{fig:layered}.

\begin{figure}
    \centering
     \includegraphics[width=.7\linewidth]{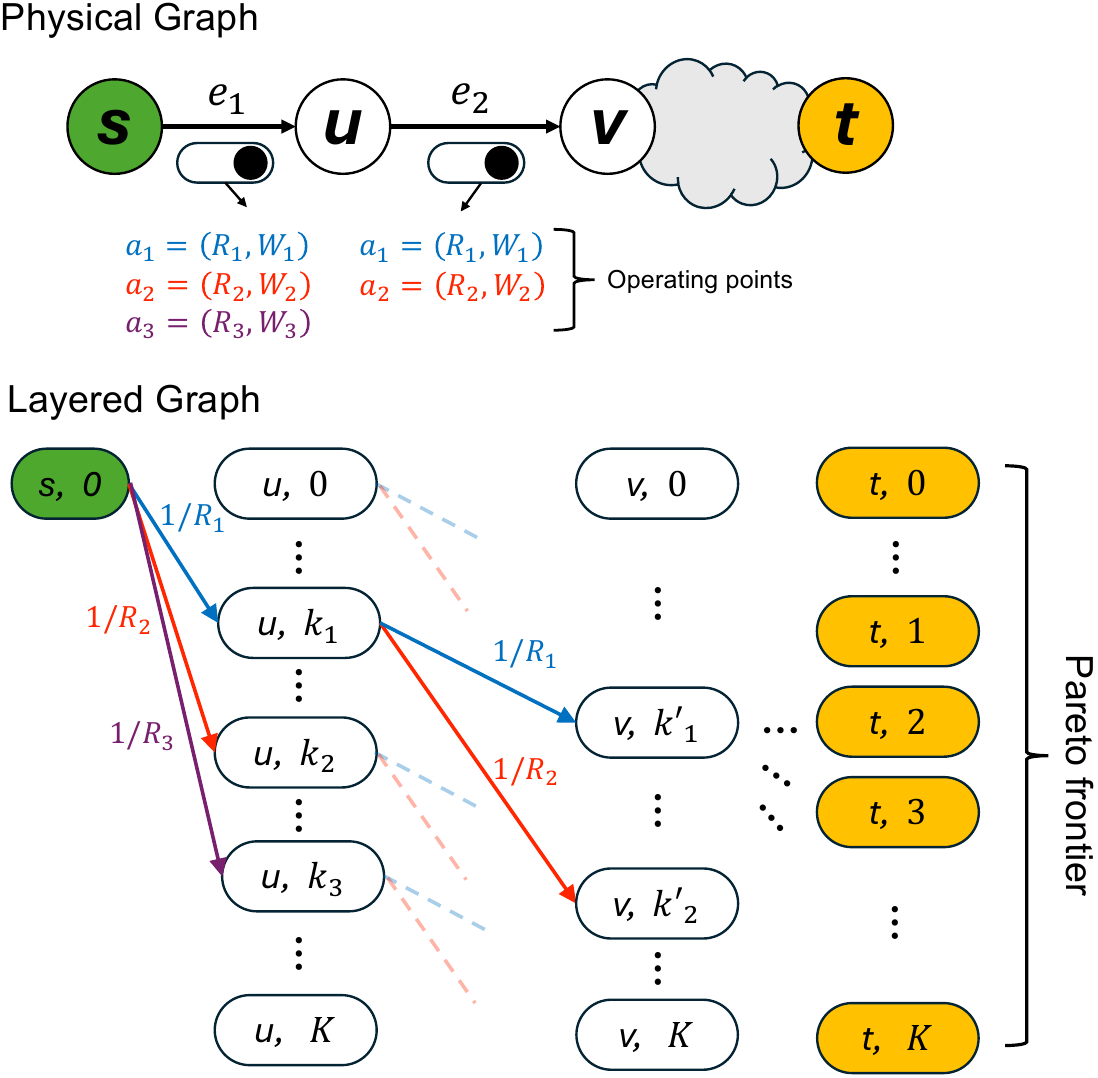}
       \vspace{-.5em}
    \caption{\textbf{Example of a physical and its associated layered graph.} In the example, the edge $e_1$ has three operating points $a_1, a_2, a_3$. In the layered graph, these correspond to three  edges from $(s,0)$ to $(u,k_1)$, $(u,k_2)$, $(u,k_3)$, where $k_i = \lceil -\ln W_i / \Delta \rceil$ with edge cost $1/R_i$ respectively. The edge $e_2$ between $u$--$v$ has instead two operating points, and thus each $(u,k)$ vertices in the layered graph have two outgoing edges corresponding to the offsets $k' = \lceil k - \ln W_i/\Delta \rceil$ and weights $1/R_i$. This construction applies to all edges and operating points of the physical graph. QLSP finds the shortest path from $(s, 0)$ to all $(t, k)$ nodes, with $0 \leq k \leq K$. In so doing, it finds the highest-rate path for all feasible end-to-end fidelities, thus the Pareto frontier.}
    \label{fig:layered}
\end{figure}

\paragraph{Interpretation}
A path $\pi=(p,a)$ in the layered graph corresponds to a physical path $p$ together with a choice $a$ of operating point on each traversed edge. The path cost in the layered graph is the cumulative inverse-rate: $X = \sum_{e\in p} 1/{R_e}$, and $Y \approx \sum_{e\in p} -\ln W_e$. Hence, if \(d(v,k)\) denotes the shortest-path distance from \((s,0)\) to \((v,k)\) in the layered graph, then $R(v,k) = 1/d(v,k)$, and $W(v,k) \approx e^{-k\Delta}$.

\paragraph{Optimization}
We run Dijkstra's algorithm on the layered graph \(G^\Delta\) from the source  $(s,0)$ which corresponds to the empty path with $(X,Y)=(0,0)$. At the destination node \(t\), each reachable layer \(k\) yields a candidate path $p_k$, and the corresponding link-level operating points: the rate $1/d(t,k)$ is already exact, while the \emph{exact} fidelity is computed as: $Y(p_k)=\sum_{e\in p_k} y_{e,a_e}$, giving $U_k = U\left(1/{d(t,k)},e^{-Y(p_k)}\right)$. The final output is the layer $k^\star = \arg\max_{k:\, d(t,k)<\infty} U_k$ together with the path $p_{k^\star}$, and utility $U_{k^\star}$.

\subsection{Algorithmic analysis}
\paragraph{Correctness}
Because all layered-edge costs are nonnegative ($1/R_{e,a}\ge 0$), 
Dijkstra's algorithm computes, for each layered state \((v,k)\), the minimum cumulative inverse-rate among all paths from \((s,0)\) to \((v,k)\). The layer index records the discretized cumulative Werner parameter, so the algorithm computes the best achievable rate for each Werner parameter bin and then selects the bin maximizing the e2e utility.

\paragraph{Complexity}\label{par:complexity}
If the physical graph has \(n=|\mathcal{V}|\) nodes and \(m=|\mathcal{E}|\) edges, and each edge has at most \(A\) operating points, then the layered graph has
$|\mathcal{V}^\Delta| = O(nK)$ nodes and $|\mathcal{E}^\Delta| = O(mKA)$ edges. Dijkstra's algorithm runs in $O\bigl(mKA\log(nK)\bigr)$ time using a standard binary heap, and uses $O(nK)$ memory up to storage of the layered edges.

\paragraph{Approximation quality}\label{par:approx-quality}
If the cumulative Werner parameter-loss variable \(Y=-\ln W\) is rounded after each edge extension using bin width \(\Delta\), then along any path of length at most \(L\), the total discretization error in \(Y\) is at most \(L\Delta\). Hence, the e2e Werner parameter satisfies $\widehat W \ge e^{-L\Delta}W$. Consider the specific case of negativity utility function (\ref{eq:neg}), which, when expressed in terms of $W$ reads $U(R, W) = R(3W-1)/4$ (we can omit the $\max$ under the assumption \(W_{\mathrm{e2e}}\ge W_{\min}>1/3\)), the approximation ratio for the negativity term \(g(W)=(3W-1)/4\) is bounded below by
\[
\frac{g(\widehat W)}{g(W)}
\ge
1-
\frac{3(1-e^{-L\Delta})}{3W-1}
\ge
1-
\frac{3L\Delta}{3W_{\min}-1}.
\]
The derivation uses the following inequalities: $W \leq 1$ and $1-e^{-x} \leq x,$ for $x\geq0$.
Thus choosing $\Delta \le {\varepsilon(3W_{\min}-1)}/{3L}$ 
is sufficient to ensure a \((1-\varepsilon)\)-approximation from Werner parameter discretization. For fixed \(W_{\min}\) bounded away from \(1/3\), this implies that we can set $ K = O\left(L/\varepsilon\right)$ giving a runtime of
$
O\left(
mA\frac{L}{\varepsilon}
\log\left(\frac{nL}{\varepsilon}\right)
\right).
$
For simple paths, \(L\le n-1\), so the runtime is polynomial in the input size and \(1/\varepsilon\), and thus QLSP is an FPTAS for the negativity-based utility~(\ref{eq:neg}). The theorem below states this for general utilities. 

\algrenewcommand\algorithmicindent{0.5em} 
\begin{algorithm}[t]
\small
\caption{The QLSP algorithm for rate--Werner utility}
\begin{algorithmic}[1]
\Require Graph \({G=(\mathcal{V},\mathcal{E})}\), source \(s\), destination \(t\), fidelity bin width \(\Delta\), edge operating points \(\{(R_{e,a},W_{e,a})\}_{a\in\mathcal{A}_e}\).
\Ensure Approximate path maximizing \(U(R_{\mathrm{e2e}}, W_{\mathrm{e2e}})\)

\State Construct layered node set
\[
\mathcal{V}^\Delta = \{(v,k): v\in \mathcal{V},\; k=0,\dots,K\}.
\]

\For{each \(e=(u,v)\in \mathcal{E}\)}
    \For{each \(a\in\mathcal{A}_e\)}
        \State Compute
        $x_{e,a}=1/{R_{e,a}},\quad y_{e,a}=-\ln W_{e,a}.$
        \For{\(k=0,\dots,K\)}
            \State Set
            $
            k'=\left\lceil \frac{k\Delta+y_{e,a}}{\Delta}\right\rceil.
            $
            \If{\(k'\le K\)}
                \State Add layered edge \((u,k)\to(v,k')\) with cost \(x_{e,a}\)
                \If{\(G\) is undirected}
                    \State Add layered edge \((v,k)\to(u,k')\) with cost \(x_{e,a}\)
                \EndIf
            \EndIf
        \EndFor
    \EndFor
\EndFor

\State Initialize distances:
\[
d(v,k)\gets \infty \quad \text{for all }(v,k)\in \mathcal{V}^\Delta,
\qquad
d(s,0)\gets 0.
\]
\State Run Dijkstra's algorithm on \(G^\Delta\) from source \((s,0)\)

\For{each \(k\in\{0,\dots,K\}\) with \(d(t,k)<\infty\)}
    \State Recover the shortest path \(p_k\) to \((t,k)\)
    \State Compute the exact fidelity loss \(Y(p_k)=\sum_{e\in p_k} y_{e,a_e}\).
    \State Compute $U_k=U\left(1/{d(t,k)},e^{-Y(p_k)}\right)$.
\EndFor

\State Return $k^\star=\arg\max_{k:\, d(t,k)<\infty} U_k$, $p_{k^\star}$, and $U_{k^\star}$.

\end{algorithmic}
\label{alg:layered-graph}
\end{algorithm}

\begin{theorem}[Approximation guarantee]\label{theorem:approximation}
Suppose that $U$ is nondecreasing in rate and fidelity and that, for a given $\epsilon >0$, there exists some $\delta>0$ such that
\[
U(R,e^{-(Y+\delta)})
\geq
(1-\epsilon)U(R,e^{-Y})
\]
for all achievable $(R,Y)$ in the relevant domain. If QLSP uses layer width $\Delta\leq{\delta}/{L}$,
then it returns a solution $\widehat\pi$ satisfying $U(\widehat\pi)\geq(1-\epsilon)U(\pi^\star)$. The number of layers is therefore $K=O(L/\delta)$ for a fixed $W_{\min} > 0$. Thus, if $1/\delta$ is polynomial in
$1/\epsilon$ and since $L\leq |\mathcal V|-1$, QLSP is an FPTAS.
\end{theorem}

\begin{proof}
The proof follows from the fact the QLSP rounds the fidelity loss of each edge upward and $\pi^*$ uses at most $L$ edges. Thus $Y_{\hat{\pi}} \leq Y_{\pi^*} + L\Delta$. To ensure $Y_{\hat{\pi}} \leq Y_{\pi^*} + \delta$, it suffices to require $L\Delta \leq \delta$. Setting $\Delta = \delta/L$ and for a fixed $W_{\textrm{min}} > 0$ the number of layers $K=\lceil{\ln (1/W_{\textrm{min}})/\Delta}\rceil=O(L/\delta)$. 
To express this bound in terms of the approximation error $\epsilon$,
we must relate $\delta$ to $\epsilon$, which depends on the form of
the utility function. For the negativity utility, we derived that
$K=O(L/\epsilon)$ as $\delta = \Theta(\epsilon)$. Similarly, for the BB84 utility, one can show that
$K=O\left(L\log(1/\epsilon)/\epsilon\right)$. These bounds assume that
$W_{\min}$ is fixed and bounded away from the corresponding
zero-utility threshold. Since $L \leq |{\cal V}|-1$, the size of the layered graph is
polynomial in $K$ and the input size, and shortest paths can be computed
in polynomial time, QLSP is an FPTAS.
\end{proof}
Fidelity-threshold utilities require additional care because an
arbitrarily small decrease in fidelity can discontinuously reduce the
utility to zero. This can be handled in two ways.

\begin{proposition}[Threshold utility]
\label{cor:threshold}
Consider $U_{\mathrm{th}}(R,W)=R\mathbf{1}[W\geq W_{\mathrm{req}}]$, where $W_{\mathrm{req}}$ is a minimum required Werner parameter of the e2e state,
and let $B=-\ln W_{\mathrm{req}}$. Suppose an optimal solution
$\pi^\star$ satisfies
$Y_{\pi^\star}\leq B-\gamma$ for some $\gamma>0$. If QLSP uses
fidelity-loss layer width $\Delta\leq\gamma/L$, then it returns an
optimal solution for $U_{\mathrm{th}}$.
Without this margin assumption, a margin-free FPTAS can instead
discretize inverse rate $X=1/R$ while accumulating fidelity loss 
$Y=-\ln W$ exactly. For any $\epsilon>0$, this rate-layered variant
returns a feasible solution $\widehat{\pi}$ satisfying
$Y_{\widehat{\pi}}\leq B$ and
$X_{\widehat{\pi}}\leq(1+\epsilon)X_{\pi^\star}$, and hence
$U_{\mathrm{th}}(\widehat{\pi})
\geq U_{\mathrm{th}}(\pi^\star)/(1+\epsilon)
\geq(1-\epsilon)U_{\mathrm{th}}(\pi^\star)$.
\end{proposition}

\subsection{Baselines}\label{subsec:baselines}
We consider several baseline approaches. Each may suffer from a mismatch between the objective optimized during path selection and the true e2e utility, although the nature of this mismatch differs across methods.

\paragraph{Sum of edge-based costs}
We first define a local edge cost using one of the following choices:
\begin{itemize}
\item  Local utility: $c_e = \min_a 1/U(R_{e,a}, W^k_{e,a})$ (with $k>0$),
\item  Rate only: $c_e= \min_a 1/R_{e,a}$
\item  Werner only: $c_e= -\max_a \ln W_{e,a}$.
\end{itemize}
We then minimize the sum of edge costs over the path. For local utility, the exponent $k$ is set to the network diameter. We found this scheme to work better than simply setting $k=1$ in our experiments---larger values of $k$ emphasize fidelity over rate. After obtaining a path \(p\), we evaluate its true e2e utility using composition rules~(\ref{eq:path_composition}).
The operating point $a_e$ for each edge is determined locally based on the minimizer of the edge cost. 

\paragraph{Scalar surrogate sweep}
An alternative approach assigns each edge the scalarized cost
\begin{equation}\label{eq:scalarization}
c_e(\alpha, \beta) = \min_{a \in \mathcal{A}_e}\Bigl(\frac{\alpha}{R_{e,a}} - \beta \ln W_{e,a}\Bigr)
\end{equation}
where $(\alpha,\beta) \geq 0$ and then runs a standard shortest-path algorithm on the graph (note: the scalar $\alpha$ here is not to be confused with the bright-state population parameter introduced in \S~\ref{sec:model}).
The surrogate matches the rate and fidelity composition rules, but is based on a \textit{scalarization} of the true objective. It can also be
viewed as a Lagrangian relaxation of the constrained problem: $\min X_\textrm{e2e}$ \emph{subject to:} $Y_\textrm{e2e} \leq B$, which maximizes rate subject to a fidelity constraint. The parameters  \((\alpha,\beta)\) must be swept over a range of values, and the resulting path be evaluated using the true e2e utility. 
The algorithm returns the path with the highest utility among all solutions found.
In Theorem~\ref{thm:main}, we show that, for utilities with convex fidelity profiles,
enumerating all supported solutions through this sweep recovers a
globally optimal solution.

We implement this sweep using an adaptive binary search. First, the coordinates \(X=1/R\) and
\(Y=-\ln W\) are normalized by the median positive per-edge contributions
\(\mu_X\) and \(\mu_Y\), respectively. Each adaptive shortest-path query therefore
minimizes \(\alpha\widetilde X+\beta\widetilde Y\), where
\((\widetilde X,\widetilde Y)=(X/\mu_X,Y/\mu_Y)\). We then solve the \emph{rate only} and
\emph{Werner only} objective, corresponding to $(\alpha, \beta) = (1,0)$ and $(0,1)$ respectively. Given two distinct solutions
\(z_\ell=(\widetilde X_\ell,\widetilde Y_\ell)\) and
\(z_r=(\widetilde X_r,\widetilde Y_r)\), where
\(\widetilde X_\ell<\widetilde X_r\), the next shortest-path query uses the
separating weights
\((\alpha,\beta)=(\widetilde Y_\ell-\widetilde Y_r,
\widetilde X_r-\widetilde X_\ell)\), for which the two solutions have equal
scalarized cost. If the query finds a new supported solution, the two
resulting intervals are searched recursively; otherwise, that interval is
exhausted. The recursion stops when no new path is found or either normalized
coordinate difference is below \(10^{-9}\). Finally, every distinct feasible path identified by the sweep is
evaluated using the true e2e utility, and the path with the
highest utility is retained. Because the scalarized costs decompose
additively over edges, for any fixed $(\alpha,\beta)$, the operating
point of each selected edge can be chosen locally using
\eqref{eq:scalarization}.

\paragraph{Distance path}
This baseline decouples physical-route selection from link
configuration. It first selects an $s$--$t$ path that minimizes total the 
physical distance, $p_{\mathrm{dist}}^* = \arg\min_{p} \sum_{e\in p} L_e,$
where $L_e$ denotes the physical length of link $e$. It then runs QLSP
on the fixed path $p_{\mathrm{dist}}^*$ using all operating
points available on its links. As shown in \S~\ref{sec:numer}, this
separation is effective when the physical route can be selected largely independently of its configuration, e.g., when homogeneous links make the shortest-distance path dominant.

\subsection{Comparison}
Across the graphs and utility functions considered in our experiments
(\S~\ref{sec:numer}), QLSP substantially outperforms baselines based on
fixed edge costs. The surrogate sweep, however, matches QLSP when the
utility function satisfies certain conditions. We formalize this result
in Theorem~\ref{thm:main}, which shows that the surrogate sweep recovers
a globally optimal solution when the fidelity profile is convex.

\begin{theorem}[Empty duality gap for convex fidelity profiles]
\label{thm:main}
Assume the utility has the form $U(R,W)=R\cdot g(W)$, with $g$ nondecreasing. Define the \emph{fidelity profile} as $\varphi(y):=g\bigl(e^{-y}\bigr)$.
Assume $\varphi$ is \emph{convex} and some solution achieves
$U>0$. Let $\pi=(p,a)$ denote a physical $s$-$t$ path $p$ with $e_k$ denoting the traversed edge $k$ on the path with a choice of operating point $a_k \in {\cal A}_{e_k}$.
Then, on \emph{every} instance the maximum
of $U$ over all
routing solutions $\pi$ is attained by a solution that also minimizes the
scalarized cost
$\sum_{k}\Bigl(\frac{\alpha}{R_{e_k,a_k}}-\beta\ln W_{e_k,a_k}\Bigr)$
for some weights $(\alpha,\beta)\ge0$. Consequently the weighted-sum sweep returns the exactly optimal utility,
provided the sweep is fine enough to hit every breakpoint of the
piecewise-linear Lagrangian
$\phi(\lambda)=\min_{\pi}\sum_{k}\bigl(1/R_{e_k,a_k}-\lambda\ln
W_{e_k,a_k}\bigr)$. There is no Lagrangian duality gap.
Conversely, when $\varphi$ is non-convex (e.g.,
$g(W)=\mathbf 1[W\ge W_{\mathrm{req}}]$, where $W_{\mathrm{req}}$ is the minimum required value, saturating ramps, sigmoids), the
utility maximizer can be a Pareto-optimal solution that no
scalarization ever returns, and
the resulting gap can be an arbitrarily large fraction of the optimum.
\end{theorem}

\begin{proof}
In the additive coordinates $(X,Y)=(1/R,-\ln W)$, the utility can be
written as $U(X,Y)=\varphi(Y)/X$. Convexity of $\varphi$ implies that
$U$ is quasiconvex, since each sublevel set
$\{(X,Y):U(X,Y)\leq c\}=\{(X,Y):X\geq\varphi(Y)/c\}$ is convex.
Therefore, some maximizer of $U$ over the convex hull of the achievable
set lies at an extreme point. Moreover, because $U$ is nonincreasing in
both $X$ and $Y$, this maximizer can be chosen on the lower-left
boundary of the convex hull. Every vertex on this boundary minimizes
$\alpha X+\beta Y$ for some $\alpha,\beta\geq0$, not both zero.
Consequently, enumerating all supported solutions by weighted-sum
scalarization recovers a globally optimal solution.
\end{proof}

\begin{remark}[Why $\varphi$ and not $g$]
Convexity is required in the additive coordinate $y=-\ln W$, because that is
the variable in which fidelity loss accumulates linearly along a path and in
which the scalarization is linear. Convexity of $g$ in $W$ (together with
monotonicity) is sufficient but not necessary: $g(W)=\sqrt W$ is concave in
$W$ yet has convex profile $\varphi(y)=e^{-y/2}$.
\end{remark}

\begin{remark}[Instances satisfying the hypothesis]
The negativity utility and the BB84 secret-key-rate utility both have
convex fidelity profiles over their positive-utility regions.
Consequently, for either utility, a scalarization sweep that enumerates
all supported solutions recovers a globally optimal solution. 
\end{remark}

\begin{remark}[Frontier computation and reuse]
\label{rem:frontier}
A single run of QLSP returns the entire
discretized rate--fidelity frontier $d(t,k)$, $k=0,\dots,K$, which can be re-evaluated under \emph{any} nondecreasing utility and for every destination
simultaneously. 
The scalarization sweep also admits a frontier interpretation: sweeping
$\alpha,\beta \geq 0$ 
enumerates the vertices of the lower convex hull of the achievable
$(R,W)$ set is also reusable. 
However, the hull omits the unsupported Pareto
points (Theorem~\ref{thm:main}); it is therefore sufficient only for
utilities with convex fidelity profiles, while the layered frontier
remains faithful for the non-convex ones (e.g., fidelity thresholds). 
The hull enumeration costs one shortest-path run per hull vertex, but the number of breakpoints of a parametric shortest path can be superpolynomial in the worst case~\cite{carstensen1983complexity}, whereas the layered run is a single computation, polynomial in the input size and $1/\varepsilon$.
\end{remark}

\begin{remark}
    Our analysis assumed a discrete set of operating points $\mathcal{A}$. In the Appendix, we propose a simple modification to our standard QLSP algorithm to accommodate continuous operating points; remarkably, provided a convexity condition on the per-link cost, one does not have to pay a $K^2$ penalty for the resulting (denser) layered graph. We show that the single-click scheme \eqref{eq:single-click-tradeoff}-\eqref{eq:single-click-operating-point} satisfies the convexity condition. Furthermore, we provide in the Appendix an argument  that access to a continuous set of operating points does not eliminate the duality gap.
\end{remark}

\begin{remark}
    Our setup and analysis assumed that $s$ and $t$ are allowed to choose a single network path. More generally, the nodes could multiplex between multiple network paths. In the Appendix, we show that multiplexing is not advantageous when the utility function is convex.
\end{remark}

\section{Numerical Evaluation}

\begin{figure*}
      \centering
      \includegraphics[width=\textwidth]{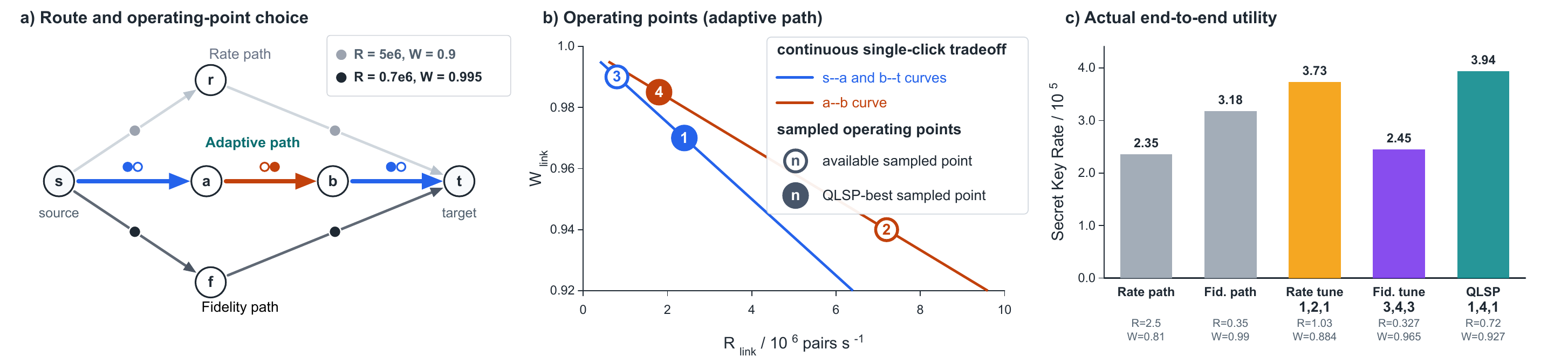}
      \vspace{-2em}
      \caption{
      \textbf{Illustration of our framework on a small network with route and operating-point choices.}
      (a) The source \(s\) and destination \(t\) are connected by a rate-seeking path (top), a fidelity-seeking path (bottom), and an adaptive path (middle) whose links have two operating points (the two circles above each edge).
      (b) The operating points on the adaptive path are drawn from linear tradeoff curves, compatible with the single-click scheme.
      Points \(1,2\) are rate-seeking points, while points \(3,4\) are fidelity-seeking points; filled markers denote the best points with respect to BB84 utility.
      (c) End-to-end SKR utility for five candidate solutions.
      Rate-seeking choices lose too much fidelity, while fidelity-seeking choices sacrifice too much rate.
      QLSP jointly optimizes the route and operating points, selecting a mixed operating points \(1,4,1\), and achieves the largest utility. Values shown in this plot are obtained by using an approximation tolerance $\epsilon = 10^{-4}$, as defined in \S~\ref{par:approx-quality}. LLEG rate and Werner parameters are calculated using the single-click scheme described in \S~\ref{sec:model}.}
      \label{fig:proof-of-concept-qnum}
      \vspace{-3mm}
  \end{figure*}

  \begin{figure}
      \centering
      \includegraphics[width=\linewidth]{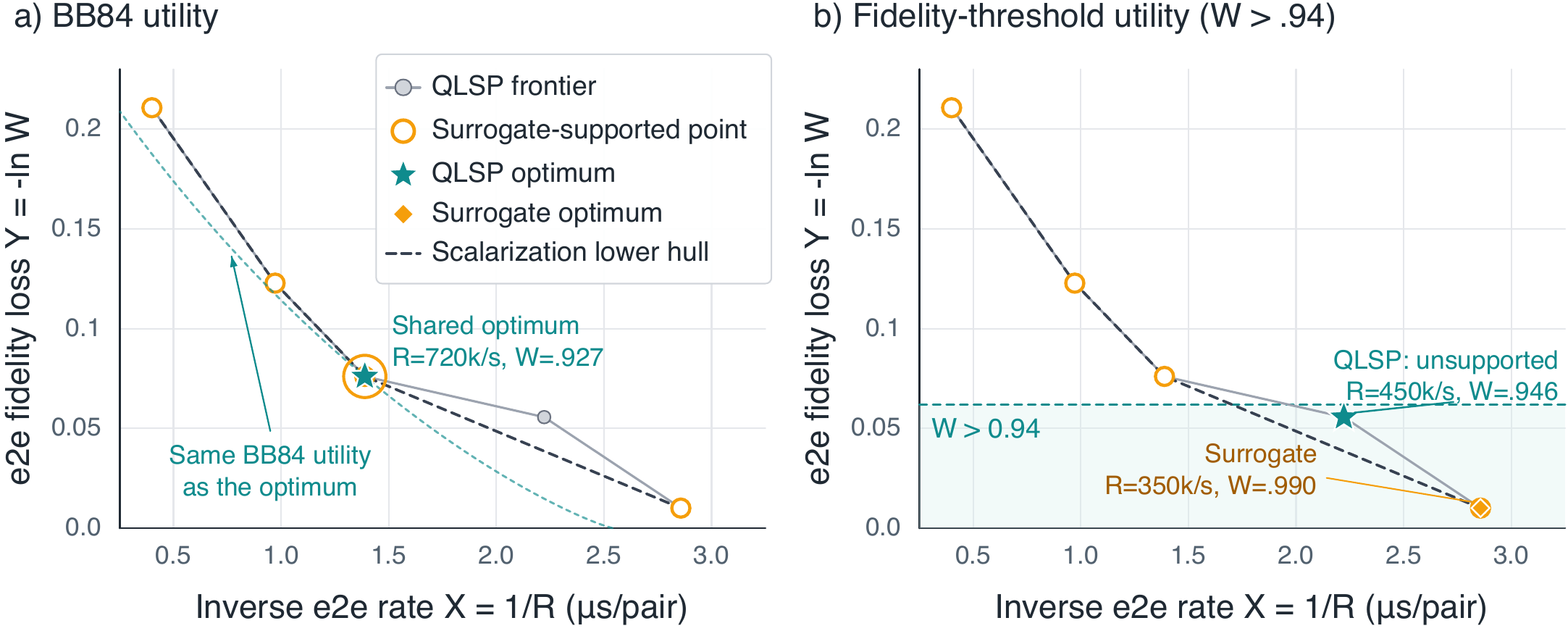}
      \caption{\textbf{Limits of scalarized routing under non-convex utilities.} The panels show the attainable end-to-end configurations in inverse-rate/fidelity-loss space; the black dashed line is the lower convex hull reachable through surrogate scalarization. (a) For BB84
      SKR, the optimum is supported, and the surrogate baseline matches QLSP. (b) Imposing the non-convex fidelity requirement $W>.94$ moves the optimum to an unsupported Pareto point. QLSP selects this configuration, whereas the surrogate is restricted to a sub-optimal,
      hull-supported solution.}
      \label{fig:proof-of-concept-nonconvex}
      \vspace{-4mm}
  \end{figure}
  
\label{sec:numer}
To evaluate QLSP, we employ two representative utility functions: the \emph{BB84 SKR} (\(U_{\text{BB84}}\)) as defined in~\eqref{eq:bb84}, and \emph{fidelity-threshold utility} ($U_{\mathrm{th}}$) as defined in~\eqref{eq:th}, where $W_{\mathrm{req}} = (4F_{\mathrm{req}}-1)/3$ is the threshold value for the Werner parameter. 
We set $W_{\mathrm{req}}=0.94$, corresponding to a $F_{\mathrm{req}}=0.955$, although the results are similar for other choices. These utilities illustrate both the convex-profile case (BB84) and the nonconvex case (fidelity threshold).

\subsection{Proof of Concept}
Fig.~\ref{fig:proof-of-concept-qnum} presents a simple example of the joint
decision solved by QLSP.  The source \(s\) and destination \(t\) are connected by three possible corridors.  The top corridor is a \emph{rate path}: it is fast
but comparatively noisy.  The bottom corridor is a \emph{fidelity path}: it is
relatively low-noise but slow.  The middle corridor is an \emph{adaptive path}: its links offer multiple operating points (two per link; Fig.~\ref{fig:proof-of-concept-qnum}b), so the optimizer must determine both the physical route and the configuration of each link.

The example illustrates two key effects. First, maximizing e2e utility can favor a route overlooked by optimizing rate or fidelity alone. A rate-only policy (e.g., \emph{rate only} baseline  from \S~\ref{subsec:baselines}), might select the rate path, while a fidelity-only policy (e.g., \emph{Werner only} baseline from \S~\ref{subsec:baselines}), might select the fidelity path. Neither maximizes $U_{\text{BB84}}$. Second, even after selecting the adaptive route, link operating points cannot be chosen independently. Configuring every adaptive link for high rate produces insufficient e2e fidelity, while configuring every link for high fidelity incurs excessive delay. Instead, the optimal solution is the mixed configuration $(1,4,1)$, which assigns a higher-fidelity operating point to the middle link and higher-rate operating points to the two side links.

Fig.~\ref{fig:proof-of-concept-nonconvex} uses the same network to show why
baselines that use scalar edge scores can be suboptimal. 
For the smooth BB84 utility, the scalarized surrogate and other baselines based on local policies can recover the same adaptive
configuration as QLSP. This is the favorable case: the optimum lies on a
supported portion of the rate--fidelity frontier.  For the threshold utility
\(U_{\textrm{th}}\) with $W_{\textrm{req}}=0.94$, however, the best solution is the mixed threshold-feasible configuration \((1,4,3)\): it provides just enough fidelity to satisfy the threshold while maximizing rate.
The baseline instead chooses the safer fidelity path and obtains a solution worse by \(22.2\%\).  The takeaway is not that this small topology is difficult, but rather that a nonconvex utility can make the optimal configuration unsupported and therefore inaccessible to scalarization-based methods. QLSP still recovers it because it explicitly maintains a discretized e2e rate--fidelity Pareto frontier.

\subsection{Comparison with Baselines}

\begin{figure*}
      \centering
      \includegraphics[width=\textwidth]{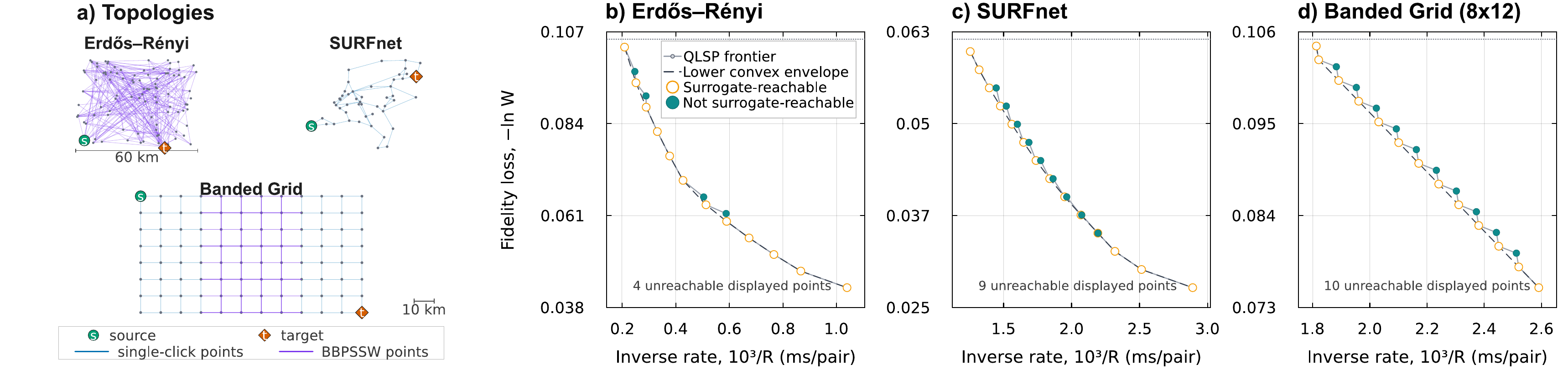}
      \vspace{-2em}
      \caption{
        \textbf{(a)} Three topologies used in the evaluation. \textbf{(b,c,d)} The rate--fidelity frontier returned by QLSP on the three topologies. In all cases, there are frontier points (green circles) that are inside the lower convex hull and cannot thus be found by the surrogate sweep. For single click operating points, we use $A=21$.}
      \label{fig:frontiers}
      \vspace{-2mm}
\end{figure*}

\begin{figure*}
    \centering
    \includegraphics[width=\textwidth]{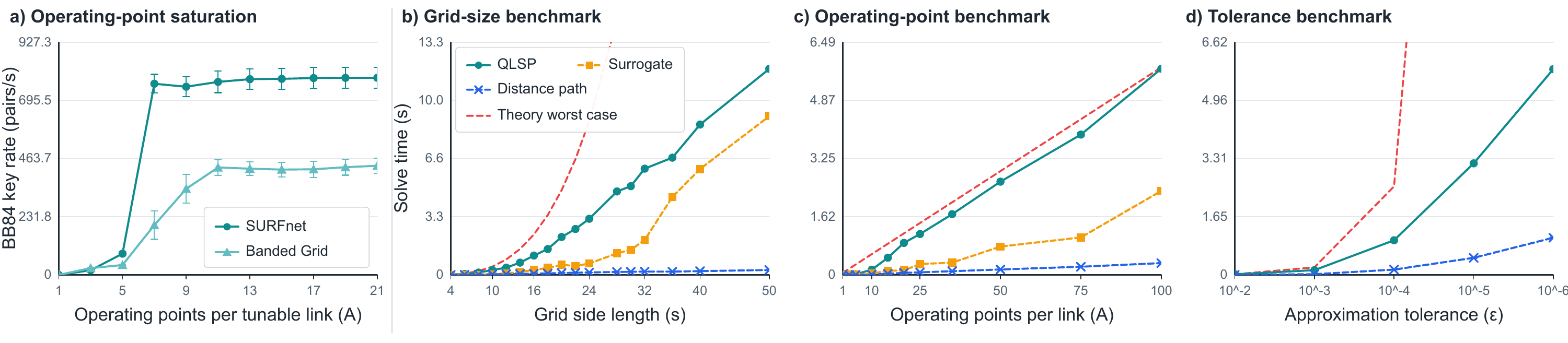}
      \vspace{-2em}
    \caption{\textbf{Operating-point resolution and runtime benchmarks.}
    (a) Expected BB84 secret key rate returned by QLSP as the number of
    single-click operating points $A$ increases; BBPSSW menus remain fixed.
    Results average 50 sampled pairs for the banded grid and 48 for SURFnet.
    Runtime is measured on square grids as a function of (b) grid side $s$,
    with $A=25$ and \(\varepsilon=10^{-3}\); (c) $A$, with $s=16$ and
    \(\varepsilon=10^{-3}\); and (d) \(\varepsilon\), with $s=8$ and
    $A=25$. Runtime experiments use negativity utility. Error bars are 95\% confidence intervals, with 32 timed
    runs per runtime point. Red curves evaluate the worst-case complexity as in
    \S~\ref{par:complexity} and are scaled as upper envelopes of the QLSP
    measurements. The surrogate is omitted from (d) because it does not use
    the layered approximation tolerance. Runtime experiments used a binary
    heap implementation on a \textit{MacBook Pro 14-inch (2024), 24GB}.}
    \label{fig:saturation-benchmark}
      \vspace{-1em}
\end{figure*}

\begin{table}[t]
\centering
\caption{Performance baselines expressed as percent of QLSP utility.}\label{tab:baseline_comp}
\scriptsize
\setlength{\tabcolsep}{3.5pt}
\renewcommand{\arraystretch}{1.12}
\begin{tabular}{@{}llccc@{}}
\hline
Topology & Utility & Local utility & Distance path & Werner only \\
\hline
Erd\"{o}s--R\'enyi & BB84 & \(69.8 \pm 11.2\) & \(\mathbf{86.3 \pm 5.5}\) & \(27.9 \pm 1.5\) \\
 & Fidelity threshold & \(57.3 \pm 1.1\) & \(\mathbf{86.8 \pm 6.0}\) & \(46.5 \pm 2.4\) \\
\hline
SURFnet & BB84 & \(80.4 \pm 5.5\) & \(\mathbf{92.8 \pm 2.1}\) & \(1.68 \pm 0.06\) \\
 & Fidelity threshold & \(69.9 \pm 8.1\) & \(\mathbf{92.7 \pm 2.1}\) & \(1.27 \pm 0.05\) \\
\hline
Banded Grid & BB84 & \(0.0 \pm 0.0\) & \(\mathbf{100.0 \pm 0.0}\) & \(5.53 \pm 0.08\) \\
 & Fidelity threshold & \(0.0 \pm 0.0\) & \(\mathbf{100.0 \pm 0.0}\) & \(22.5 \pm 3.7\) \\
\hline
\end{tabular}
\end{table}

We evaluate the algorithms on three representative networks (Fig.~\ref{fig:frontiers}a): a $100$-node Erd\"{o}s--Rényi (ER) graph~\cite{erdos1959} with $p=0.05$, the $50$-node \emph{SURFnet} real-life topology~\cite{Topology_zoo}, and an $8 \times 12$ grid with homogeneous 10 km links. For the ER topology, nodes are randomly placed on a $60$ km square, and link lengths are set to the Euclidean distance between nodes. 

As shown in Fig.~\ref{fig:frontiers}a, all SURFnet links use the
single-click model in \eqref{eq:single-click-tradeoff}--\eqref{eq:single-click-operating-point}.  We set the LLEG repetition rate to
$R_{\mathrm{rep}}=10^7~\mathrm{s}^{-1}$, the aggregate non-fiber efficiency
to $c=0.36$, and the fiber loss to $0.2~\mathrm{dB/km}$, equivalently
$L_{\mathrm{att}}\simeq22~\mathrm{km}$.  These parameters determine the rate
scale but do not affect the qualitative comparisons.  Each tunable
single-click link is represented by $A$ operating points sampled uniformly over
Werner parameter $W\in[0.9,0.9999]$.  We vary $A$ to measure how performance changes as the continuous physical tradeoff is represented more finely.

In the ER topology, single-click generation is fixed at $W_0=0.96$ and its
corresponding rate $R_0$.  Each link instead exposes the five-point BBPSSW menu
obtained from 0-4 distillation rounds using
\eqref{eq:bbpssw-success} and \eqref{eq:bbpssw-recurrence}.  The
protocol-banded grid isolates hardware heterogeneity while keeping every link
length equal to $10$ km: links in the outer regions use the tunable
single-click menu, whereas links within central columns 5--8, together with
the interfaces entering and leaving this region, use the same five-point
BBPSSW menu (purple in Fig.~\ref{fig:frontiers}a).

\paragraph{Discussion}\label{par:discussion}
The proof-of-concept in Fig.~\ref{fig:proof-of-concept-nonconvex}
established how scalarization omits solutions favored by a non-convex
utility. Fig.~\ref{fig:frontiers} examines whether the same phenomenon persists
in representative networks. In all three settings, the QLSP frontier contains
Pareto-optimal points that are not supported by its lower convex envelope.
Thus, unsupported configurations arise with spatially heterogeneous links, in
a real network topology, and with heterogeneous protocols. This
comparison is independent of a particular utility: every displayed unsupported
point is a feasible candidate that may be preferred by an appropriate monotone
non-convex utility, but that no choice of surrogate weights can return. 

Having isolated this structural limitation of the surrogate, we compare in Table~\ref{tab:baseline_comp}
QLSP with three heuristic baselines---\emph{Local utility}, \emph{distance path}, and \emph{Werner only}---from \S~\ref{subsec:baselines}.
For a baseline \(b\), each entry reports the performance ratio
\begin{equation}
    \rho_b =
    100\,\mathbb{E}_{(s,t)}
    \left[
        U_b(s,t)/U_{\mathrm{QLSP}}(s,t)
    \right],
\end{equation}
where the expectation is over all $s$--$t$ pairs separated by at least eight hops. A value of \(100\%\) therefore
means that the baseline always matches QLSP, while lower values measure the
average fraction of QLSP utility retained. Values are obtained as an average over $50$ random pairs and errors are $95$\% confidence intervals. We did not include \emph{rate only} since it disregards fidelity entirely and its performance never matches that of QLSP.

The \emph{distance path} is the strongest of these baselines, retaining approximately
\(86\%\)--\(93\%\) of QLSP utility on ER and SURFnet and matching QLSP on the
banded grid. Selecting a physical path before configuring it can be lossless
when the two decisions are separable, for example when only one path exists,
when the few available paths can all be configured and compared, or when
homogeneity makes a shortest path dominate independently of its
configuration. The banded grid is favorable in this sense: its regular geometry makes the distance path sufficient for the sampled pairs. In general, heterogeneous link lengths or operating-point
menus can make a physically longer path preferable after configuration, as the
ER and SURFnet gaps demonstrate. The more local policies, as expected, are less robust:
\emph{Local utility} retains between \(0\%\) and approximately \(80\%\), while
\emph{Werner only} retains approximately \(1\%\)--\(47\%\). Their
performance also changes between BB84 and threshold utility, showing that a
fixed local objective need not remain aligned with the application objective.

\paragraph{Saturation and Runtime}\label{par:benchmark}

Fig.~\ref{fig:saturation-benchmark}a evaluates how finely the continuous
single-click tradeoff must be represented. A small $A$ can exclude useful
configurations and substantially reduce the achievable key rate. SURFnet
enters its high-utility regime around \(A=7\), whereas the banded grid
stabilizes around \(A=11\); beyond these values, additional operating points
provide only marginal gains. The resolution required
for saturation depends on the topology and protocol composition. 
We observe similar trends across additional topologies and utility functions, omitted here for space.

Fig.~\ref{fig:saturation-benchmark}b show that QLSP is practical at the network
sizes considered here. In particular, it solves the largest \(50\times50\)
grid, with \(2{,}500\) nodes and \(4{,}900\) links, in approximately
\(11.8\) seconds on a \textit{2024 MacBook Pro}. The observed growth is consistent with the polynomial worst-case analysis in \S~\ref{par:complexity}. 
Runtime grows approximately linearly with the number of operating points \(A\) (Fig.~\ref{fig:saturation-benchmark}c), as expected from the explicit \(A\) factor in the complexity bound. Decreasing the approximation tolerance (Fig.~\ref{fig:saturation-benchmark}d) increases the number of fidelity layers and, consequently, the explored state space; nevertheless, QLSP solves the
\(\varepsilon=10^{-6}\) instance on the fixed \(8\times8\) grid in
approximately \(5.9\) seconds. The theoretical curves are scaled worst-case upper envelopes rather than regression fits, so their comparison is only intended to illustrate the asymptotic trend.

The other methods reduce computational cost by searching a smaller solution
space. 
The adaptive surrogate sweep is generally faster than QLSP, although enumerating all supported solutions may require a superpolynomial number of shortest-path queries in the worst case.
The \emph{distance path} restriction yields the largest reduction: on the
\(50\times50\) grid it requires approximately \(0.28\) seconds because QLSP is
run only on the preselected physical path. This makes path preselection
valuable when it is independent of link configuration. Outside these
regimes, the runtime advantage must be weighed against the utility losses reported in Table~\ref{tab:baseline_comp}, since fixing the physical path can exclude the globally optimal joint route and configuration.

\section{Conclusion}
  \label{sec:concl}
  We have presented and analyzed QLSP---an efficient approximation algorithm for jointly selecting and configuring paths in a quantum network so as to maximize a quantum utility objective between two network nodes. We have also characterized when a potentially computationally cheaper alternative scheme---based on a weighted-sum scalarization of functions of entanglement generation rate and fidelity---is optimal. Through numerical evaluation on a variety of quantum networks, we have shown that QLSP is able to reach higher quantum utility values compared to simpler baselines that track only rate, fidelity, or a local link-level utility. 
  As future directions, QLSP could be adapted to accommodate bipartite entanglement beyond Werner states, multi-path entanglement routing (i.e., a multi-commodity formulation of the problem), and non-Pauli noise models.
  \section*{Use of AI Disclosure}
  \label{sec:disclosure}
The authors used \textit{ChatGPT 5.6}, under their supervision, to assist with documenting the code, orchestrating automated runs and collecting results, and preparing plot layouts and annotations. The authors also used \textit{ChatGPT 5.6} and \textit{Claude Opus 5} to improve the clarity of the manuscript. The authors reviewed and take responsibility for all generated content, analyses, and results.

\section*{Acknowledgment}
 This work was supported in part by the NSF award \#2522101. It is also supported in part by the NSF grants  \#2346089,  \#2402861, and NSF- ERC Center for Quantum Networks grant EEC-1941583.



\bibliographystyle{IEEEtran}
\bibliography{reference}

@article{caleffi2017optimal,
  title={Optimal routing for quantum networks},
  author={Caleffi, Marcello},
  journal={IEEE Access},
  volume={5},
  pages={22299--22312},
  year={2017},
  publisher={IEEE}
}

@article{ghaderibaneh2022efficient,
  title={Efficient quantum network communication using optimized entanglement swapping trees},
  author={Ghaderibaneh, Mohammad and Zhan, Caitao and Gupta, Himanshu and Ramakrishnan, CR},
  journal={IEEE TQE},
  volume={3},
  pages={1--20},
  year={2022},
  publisher={IEEE}
}

@article{shi2020concurrent,
  title={Concurrent Entanglement Routing for Quantum Networks: Model and Designs},
  author={Shi, Shouqian and Zhang, Xiaoxue and Qian, Chen},
  journal={IEEE/ACM Transactions on Networking},
  volume={32},
  number={3},
  pages={2205--2220},
  year={2024},
  publisher={IEEE}
}

@article{chakraborty2020entanglement,
  title={Entanglement distribution in a quantum network: A multicommodity flow-based approach},
  author={Chakraborty, Kaushik and Elkouss, David and Rijsman, Bruno and Wehner, Stephanie},
  journal={IEEE TQE},
  volume={1},
  pages={1--21},
  year={2020},
  publisher={IEEE}
}

@article{victora2023entanglement,
  title={Entanglement purification on quantum networks},
  author={Victora, Michelle and Tserkis, Spyros and Krastanov, Stefan and de la Cerda, Alexander Sanchez and Willis, Steven and Narang, Prineha},
  journal={Physical Review Research},
  volume={5},
  number={3},
  pages={033171},
  year={2023},
  publisher={APS}
}

@article{NUM1,
  author={F. P. Kelly and K. Aman and Maulloo and D. K. H. Tan},
  journal={Journal of the
Operational Research Society}, 
  title={Rate control for communication
networks: shadow prices, proportional fairness and stability.}, 
  year={1998},
  volume={49},
  number={3},
  pages={237-252},
  doi={}}

@article{NUM2,
  author={S. H. Low and D. E. Lapsley},
  journal={IEEE/ACM Transactions on Networking,}, 
  title={{Optimization flow control. I. Basic algorithm and
convergence.}}, 
  year={1999},
  volume={7},
  number={6},
  pages={861-874},
  doi={}}

@article{NUM3,
  author={D. P. Palomar and M. Chiang},
  journal={IEEE Trans. on Automatic
Control}, 
  title={Alternative distributed algorithms for network
utility maximization: Framework and applications.}, 
  year={2007},
  volume={52},
  number={12},
  pages={2254-2269},
  doi={}}

@article{MAT2015Qrepeaters,
  author={Munro, William J. and Azuma, Koji and Tamaki, Kiyoshi and Nemoto, Kae},
  journal={IEEE Journal of Selected Topics in Quantum Electronics}, 
  title={Inside Quantum Repeaters}, 
  year={2015},
  volume={21},
  number={3},
  pages={78-90},
  doi={10.1109/JSTQE.2015.2392076}}

@misc{SWarxivQoS,
  doi = {10.48550/ARXIV.2111.13124},
  
  url = {https://arxiv.org/abs/2111.13124},
  
  author = {Skrzypczyk, Matthew and Wehner, Stephanie},
  
  title = {An Architecture for Meeting Quality-of-Service Requirements in Multi-User Quantum Networks},
  
  publisher = {arXiv},
  
  year = {2021}}

@article{2012_DiFranco_OptimalPath,
  title = {Optimal path for a quantum teleportation protocol in entangled networks},
  author = {Di Franco, C. and Ballester, D.},
  journal = {Phys. Rev. A},
  volume = {85},
  issue = {1},
  pages = {010303},
  numpages = {4},
  year = {2012},
  month = {Jan},
  publisher = {American Physical Society},
  doi = {10.1103/PhysRevA.85.010303}
}

@article{2013_vanMeter_path_selection,
title = {Path selection for quantum repeater networks},
author ={Van Meter, R. and Satoh, T. and Ladd, T. D. and Munro, B. and Nemoto, K.},
journal ={Networking Science},
year = {2013},
volume = {3},
pages ={82--95},
doi = {10.1007/s13119-013-0026-2},
publisher = {Tsinghua University Press}
}

@article{2014.NatPhys.Komar-Lukin.QuantClocks,
archivePrefix = {arXiv},
arxivId = {arXiv:1310.6045v1},
author = {K{\'{o}}m{\'{a}}r, Peter and Kessler, Eric M. and Bishof, Michael and Jiang, Liang and S{\o}rensen, Anders S. and Ye, Jun and Lukin, Mikhail D.},
doi = {10.1038/nphys3000},
eprint = {arXiv:1310.6045v1},
isbn = {doi:10.1038/nphys3000},
issn = {1745-2473},
journal = {Nature Physics},
month = {oct},
number = {8},
pages = {582--587},
title = {{A quantum network of clocks}},
volume = {10},
year = {2014}
}

@article{2016.ArXiv.Schoute-Wehner.PerfQnetroute,
  title={Shortcuts to quantum network routing},
  author={Schoute, Eddie and Mancinska, Laura and Islam, Tanvirul and Kerenidis, Iordanis and Wehner, Stephanie},
  journal={arXiv preprint arXiv:1610.05238},
  year={2016}
}

@article{2014.TheoCompSc.Bennett-Brassard.BB84,
author = {Bennett, Charles H. and Brassard, Gilles},
doi = {10.1016/j.tcs.2014.05.025},
issn = {03043975},
journal = {Theoretical Computer Science},
pages = {7--11},
title = {{Quantum cryptography: Public key distribution and coin tossing}},
volume = {560},
year = {2014}
}

@article{2012.PRL.Gottesman-Croke.QRepTelescope,
author = {Gottesman, Daniel and Jennewein, Thomas and Croke, Sarah},
doi = {10.1103/PhysRevLett.109.070503},
issn = {0031-9007},
journal = {Physical Review Letters},
month = {aug},
number = {7},
pages = {070503},
publisher = {American Physical Society},
title = {{Longer-Baseline Telescopes Using Quantum Repeaters}},
volume = {109},
year = {2012}
}

@article{1991.PRL.Ekert.QCrypt,
author = {Ekert, Artur K.},
doi = {10.1103/PhysRevLett.67.661},
issn = {0031-9007},
journal = {Physical Review Letters},
month = {aug},
number = {6},
pages = {661--663},
publisher = {American Physical Society},
title = {{Quantum cryptography based on Bell's theorem}},
volume = {67},
year = {1991}
}

@inproceedings{bb84,
	address = {New York},
	author = {Bennett, C. H. and Brassard, G.},
	booktitle = {Proc. IEEE Int. Conf. Comput., Syst. Signal Process.},
	location = {Bangalore, India},
	pages = {175--179},
	publisher = {IEEE Press},
	title = {{Quantum Cryptography: Public Key Distribution and Coin Tossing}},
	year = {1984}
}

@article{pant2019routing,
  title={Routing entanglement in the quantum internet},
  author={Pant, Mihir and Krovi, Hari and Towsley, Don and Tassiulas, Leandros and Jiang, Liang and Basu, Prithwish and Englund, Dirk and Guha, Saikat},
  journal={npj Quantum Information},
  volume={5},
  number={1},
  pages={25},
  year={2019},
  publisher={Nature Publishing Group UK London}
}

@article{vardoyan2024bipartite,
  title={On the bipartite entanglement capacity of quantum networks},
  author={Vardoyan, Gayane and Van Milligen, Emily and Guha, Saikat and Wehner, Stephanie and Towsley, Don},
  journal={IEEE Transactions on Quantum Engineering},
  volume={5},
  pages={1--14},
  year={2024},
  publisher={IEEE}
}

@article{a10,
  author    = {Li, Jian and Wang, Mingjun and Xue, Kaiping and Li, Ruidong and Yu, Nenghai and Sun, Qibin and Lu, Jun},
  title     = {Fidelity-Guaranteed Entanglement Routing in Quantum Networks},
  journal   = {IEEE Transactions on Communications},
  volume    = {70},
  number    = {10},
  year      = {2022},
  pages     = {6748-6763}
}

@article{giovannetti2011advances,
  title={Advances in quantum metrology},
  author={Giovannetti, Vittorio and Lloyd, Seth and Maccone, Lorenzo},
  journal={Nature photonics},
  volume={5},
  number={4},
  pages={222--229},
  year={2011},
  publisher={Nature Publishing Group UK London}
}

@article{khabiboulline2019quantum,
  title={Quantum-assisted telescope arrays},
  author={Khabiboulline, Emil T and Borregaard, Johannes and De Greve, Kristiaan and Lukin, Mikhail D},
  journal={Physical Review A},
  volume={100},
  number={2},
  pages={022316},
  year={2019},
  publisher={APS}
}

@article{jiang2007distributed,
  title={Distributed quantum computation based on small quantum registers},
  author={Jiang, Liang and Taylor, Jacob M and S{\o}rensen, Anders S and Lukin, Mikhail D},
  journal={Physical Review A—Atomic, Molecular, and Optical Physics},
  volume={76},
  number={6},
  pages={062323},
  year={2007},
  publisher={APS}
}

@article{cabrillo1999creation,
  title={Creation of entangled states of distant atoms by interference},
  author={Cabrillo, Carlos and Cirac, J Ignacio and Garcia-Fernandez, Pablo and Zoller, Peter},
  journal={Physical Review A},
  volume={59},
  number={2},
  pages={1025},
  year={1999},
  publisher={APS}
}

@inproceedings{vardoyan2023quantum,
  title={Quantum network utility maximization},
  author={Vardoyan, Gayane and Wehner, Stephanie},
  booktitle={2023 IEEE International Conference on Quantum Computing and Engineering (QCE)},
  volume={1},
  pages={1238--1248},
  year={2023},
  organization={IEEE}
}

@article{deutsch1996quantum,
  title={Quantum privacy amplification and the security of quantum cryptography over noisy channels},
  author={Deutsch, David and Ekert, Artur and Jozsa, Richard and Macchiavello, Chiara and Popescu, Sandu and Sanpera, Anna},
  journal={Physical review letters},
  volume={77},
  number={13},
  pages={2818},
  year={1996},
  publisher={APS}
}

@article{dur1999quantum,
  title={Quantum repeaters based on entanglement purification},
  author={D{\"u}r, W and Briegel, H-J and Cirac, JI and Zoller, P},
  journal={Phys. Rev. A},
  volume={59},
  number={1},
  pages={169},
  year={1999},
  publisher={APS}
}

@article{bennett1996purification,
  title={Purification of noisy entanglement and faithful teleportation via noisy channels},
  author={Bennett, Charles H and Brassard, Gilles and Popescu, Sandu and Schumacher, Benjamin and Smolin, John A and Wootters, William K},
  journal={Physical review letters},
  volume={76},
  number={5},
  pages={722},
  year={1996},
  publisher={APS}
}

@book{nielsen2010quantum,
  title={Quantum computation and quantum information},
  author={Nielsen, Michael A and Chuang, Isaac L},
  year={2010},
  publisher={Cambridge university press}
}

@article{Bennett96,
  title = {Mixed-state entanglement and quantum error correction},
  author = {Bennett, Charles H. and DiVincenzo, David P. and Smolin, John A. and Wootters, William K.},
  journal = {Phys. Rev. A},
  volume = {54},
  issue = {5},
  pages = {3824--3851},
  numpages = {0},
  year = {1996},
  month = {Nov},
  publisher = {American Physical Society},
  doi = {10.1103/PhysRevA.54.3824},
}

@article{vidal2002computable,
  title={Computable measure of entanglement},
  author={Vidal, Guifr{\'e} and Werner, Reinhard F},
  journal={Physical Review A},
  volume={65},
  number={3},
  pages={032314},
  year={2002},
  publisher={APS}
}

@inproceedings{zhao2022e2e,
  title={{E2E fidelity aware routing and purification for throughput maximization in quantum networks}},
  author={Zhao, Yangming and Zhao, Gongming and Qiao, Chunming},
  booktitle={IEEE INFOCOM 2022},
  pages={480--489},
  year={2022}
}

@article{xiao2024purification,
  title={Purification scheduling control for throughput maximization in quantum networks},
  author={Xiao, Zirui and Li, Jian and Xue, Kaiping and Yu, Nenghai and Li, Ruidong and Sun, Qibin and Lu, Jun},
  journal={Communications Physics},
  volume={7},
  number={1},
  pages={307},
  year={2024},
  publisher={Nature Publishing Group UK London}
}

@article{fitzsimons2017unconditionally,
  title={Unconditionally verifiable blind quantum computation},
  author={Fitzsimons, Joseph F and Kashefi, Elham},
  journal={Phys. Rev. A},
  volume={96},
  number={1},
  pages={012303},
  year={2017},
  publisher={APS}
}

@article{panigrahy2025framework,
  title={A Framework for Distributed Resource Allocation in Quantum Networks},
  author={Panigrahy, Nitish K and Bacciottini, Leonardo and Hollot, CV and Van Milligen, Emily A and de Andrade, Matheus Guedes and Rao, Nageswara SV and Vardoyan, Gayane and Towsley, Don},
  journal={arXiv preprint arXiv:2510.09371},
  year={2025}
}

@inproceedings{kar2026utility,
  title={{On Utility-optimal Entanglement Routing in Quantum Networks}},
  author={Kar, Sounak and Mukhopadhyay, Arpan},
  booktitle={2026 International Conference on Quantum Communications, Networking, and Computing (QCNC)},
  pages={457--464},
  year={2026},
  organization={IEEE}
}

@inproceedings{juttner2001lagrange,
  author    = {J{\"u}ttner, Alp{\'a}r and Szviatovszki, Bal{\'a}zs and M{\'e}cs, Ildik{\'o} and Rajk{\'o}, Zsolt},
  title     = {Lagrange Relaxation Based Method for the {QoS} Routing Problem},
  booktitle = {Proc. IEEE INFOCOM},
  year      = {2001},
  pages     = {859--868}
}

@phdthesis{carstensen1983complexity,
  author = {Carstensen, Patricia J.},
  title  = {The Complexity of Some Problems in Parametric Linear and Combinatorial Programming},
  school = {University of Michigan},
  year   = {1983}
}

@article{hassin1992approximation,
  author  = {Hassin, Refael},
  title   = {Approximation Schemes for the Restricted Shortest Path Problem},
  journal = {Mathematics of Operations Research},
  volume  = {17},
  number  = {1},
  pages   = {36--42},
  year    = {1992}
}

@article{wang1996quality,
  author  = {Wang, Zheng and Crowcroft, Jon},
  title   = {Quality-of-Service Routing for Supporting Multimedia Applications},
  journal = {IEEE Journal on Selected Areas in Communications},
  volume  = {14},
  number  = {7},
  pages   = {1228--1234},
  year    = {1996}
}

@article{lorenz2001simple,
  author  = {Lorenz, Dean H. and Raz, Danny},
  title   = {A Simple Efficient Approximation Scheme for the Restricted Shortest Path Problem},
  journal = {Operations Research Letters},
  volume  = {28},
  number  = {5},
  pages   = {213--219},
  year    = {2001}
}

@article{kwiat1999ultrabright,
  title={Ultrabright source of polarization-entangled photons},
  author={Kwiat, Paul G and Waks, Edo and White, Andrew G and Appelbaum, Ian and Eberhard, Philippe H},
  journal={Physical Review A},
  volume={60},
  number={2},
  pages={R773},
  year={1999},
  publisher={APS}
}

@article{abane2025entanglement,
  title={Entanglement routing in quantum networks: A comprehensive survey},
  author={Abane, Amar and Cubeddu, Michael and Mai, Van Sy and Battou, Abdella},
  journal={IEEE Transactions on Quantum Engineering},
  year={2025},
  publisher={IEEE}
}

@article{coutinho2024fidelitycurves,
  author  = {Coutinho, Bruno C. and Monteiro, Raul and Bugalho, Lu\'is and Monteiro, Francisco A.},
  title   = {Entanglement Routing Based on Fidelity Curves},
  journal = {arXiv preprint arXiv:2303.12864},
  year    = {2024}
}

@inproceedings{zhang2025linkconfig,
  author    = {Zhang, Qiaolun and Di Cicco, Nicola and Ibrahimi, M\"emedhe and Almeida Jr., Raul C. and Gatto, Alberto and Boutaba, Raouf and Tornatore, Massimo},
  title     = {Link Configuration for Fidelity-Constrained Entanglement Routing in Quantum Networks},
  booktitle = {IEEE INFOCOM 2025 -- IEEE Conference on Computer Communications},
  year      = {2025}
}

@article{erdos1959,
  title={{On random graphs I}},
  author={Erd{\H{o}}s, Paul and R{\'e}nyi, Alfr{\'e}d},
  journal={Publicationes Mathematicae Debrecen},
  volume={6},
  pages={290--297},
  year={1959}
}

@article{Topology_zoo,
  author={Knight, Simon and Nguyen, Hung X. and Falkner, Nickolas and Bowden, Rhys and Roughan, Matthew},
  journal={IEEE Journal on Selected Areas in Communications}, 
  title={The Internet Topology Zoo}, 
  year={2011},
  volume={29},
  number={9},
  pages={1765-1775},
  doi={10.1109/JSAC.2011.111002}}

@article{dljkstra1959note,
  title={A Note on Two Problems in Connexion with Graphs},
  author={Dijkstra, EW},
  journal={Numerische Mathematik},
  volume={50},
  pages={269--271},
  year={1959}
}

@article{werner1989quantum,
  title={{Quantum states with Einstein-Podolsky-Rosen correlations admitting a hidden-variable model}},
  author={Werner, Reinhard F},
  journal={Physical Review A},
  volume={40},
  number={8},
  pages={4277},
  year={1989},
  publisher={APS}
}

@inproceedings{aggarwal1986geometric,
  title={Geometric applications of a matrix searching algorithm},
  author={Aggarwal, Alok and Klawe, Maria and Moran, Shlomo and Shor, Peter and Wilber, Robert},
  booktitle={Proceedings of the second annual symposium on Computational geometry},
  pages={285--292},
  year={1986}
}
%

\appendix
In this appendix, we show that:
1) time multiplexing between multiple network paths is not advantageous over using exclusively one path, when the utility is convex; 2) continuous operating points do not eliminate the duality gap, so although scalarization allows us to incorporate continuous operating points easily, it is not guaranteed to recover the optimum; and 3) under certain conditions, incorporating continuous operating points into the layered graph can be done without paying the $K^2$ penalty. The basic idea is to switch to a hop-indexed dynamic program (DP), similar to Bellman--Ford, and exploit a special property of the transition structure to compute all of the min-plus values efficiently. 
When the transition matrix has the so-called Monge structure, one can show that the indices of the optimal transitions vary monotonically across layers. This allows all optimal values to be computed theoretically in $O(K)$ time, or in $O(K\log K)$ time using a simpler divide-and-conquer scheme. 

\section*{Time multiplexing}\label{app:time-muxing}
Generally, the network may time-multiplex among multiple routes or configurations.
Consider two routing solutions with end-to-end rate--fidelity pairs
$(R_1,F_1)$ and $(R_2,F_2)$. Suppose that the first solution is used for
a fraction $f\in[0,1]$ of the time. The aggregate generation rate is
\begin{equation}
    R(f)=fR_1+(1-f)R_2.
    \label{eq:mux-rate}
\end{equation}
Because the two solutions generate $fR_1$ and $(1-f)R_2$ pairs per unit
time, respectively, the fidelity of the average distributed state is
\begin{equation}
    F(f)
    =
    \frac{fR_1F_1+(1-f)R_2F_2}
         {fR_1+(1-f)R_2}.
    \label{eq:mux-fidelity}
\end{equation}

Consider a utility of the form $U(R,F)=R\,g(F)$, where $g$ is convex. Let $\lambda(f) = fR_1/(fR_1+(1-f)R_2)$. By Jensen's inequality,
\begin{align}
    U(R(f),F(f))
    &=
    R(f)g\!\left(
        \lambda(f)F_1+(1-\lambda(f))F_2
    \right) \nonumber\\
    &\leq
    R(f)\left[
        \lambda(f)g(F_1)
        +(1-\lambda(f))g(F_2)
    \right] \nonumber\\
    &=
    fR_1g(F_1)+(1-f)R_2g(F_2) \nonumber\\
    &=
    fU(R_1,F_1)+(1-f)U(R_2,F_2) \nonumber\\
    &\leq
    \max\{U(R_1,F_1),U(R_2,F_2)\}.
    \label{eq:mux-convex}
\end{align}
Thus, time multiplexing cannot outperform the better of the two solutions when $g$ is convex. This
includes the negativity and BB84 utilities. Time multiplexing can, however, improve a nonconvex utility. Consider
\[
U_{\mathrm{th}}(R,F)
=
R\mathbf{1}[F\geq F_{\mathrm{req}}].
\]
Suppose solution 1 has higher rate but falls below the threshold, while
solution 2 is feasible:
\[
R_1>R_2,
\qquad
F_1<F_{\mathrm{req}}\leq F_2.
\]
The multiplexed solution satisfies the fidelity requirement if and only if
\begin{equation}
    fR_1(F_1-F_{\mathrm{req}})
    +(1-f)R_2(F_2-F_{\mathrm{req}})
    \geq 0.
    \label{eq:mux-threshold-condition}
\end{equation}
Since $R(f)$ increases with $f$, the optimal mixture uses the largest
feasible fraction of the high-rate solution. The threshold is therefore
active, yielding
\begin{equation}
    f^\star
    =
    \frac{R_2(F_2-F_{\mathrm{req}})}
    {R_1(F_{\mathrm{req}}-F_1)
     +R_2(F_2-F_{\mathrm{req}})}.
    \label{eq:mux-optimal-fraction}
\end{equation}
The resulting rate is
\begin{equation}
    R^\star_{\mathrm{mux}}
    =
    \frac{
        R_1R_2(F_2-F_1)
    }{
        R_1(F_{\mathrm{req}}-F_1)
        +R_2(F_2-F_{\mathrm{req}})
    }.
    \label{eq:mux-optimal-rate}
\end{equation}
Whenever $F_2>F_{\mathrm{req}}$, we have $f^\star>0$, and hence
\[
R^\star_{\mathrm{mux}}
=
f^\star R_1+(1-f^\star)R_2
>
R_2.
\]
Thus, mixing a high-rate, sub-threshold solution with a lower-rate,
high-fidelity solution can strictly outperform the feasible solution
used alone. 

\section*{Continuous Operating Points}
\label{app:continuous}
Continuous operating points \textit{do not} eliminate the global scalarization gap. Even when the configuration problem on every fixed path is convex and has zero duality gap, the full routing problem includes a discrete choice among physical paths. The global achievable set is a union of path-specific feasible sets and need not be convex. Moreover, each path may require a different Lagrange multiplier. Thus, a configuration can be optimal for the constrained problem yet unsupported by every global weighted-sum scalarization.

Fig.~\ref{fig:continuous-scalarization-gap} illustrates this reasoning. The network in this example contains
three internally disjoint two-link paths $P_1$, $P_2$, and $P_3$ between
$s$ and $t$. 
Solutions are shown from coarse to relatively fine operating point tuning (as the plots go from left to right).
In the $(X,Y)=(1/R,-\ln W)$ plane, the configurations of each
physical path form a cluster of points. With continuous link tuning, these
clusters extend into path-specific attainable regions.
Such regions can overlap; they happen to be disjoint here. Increasing the
resolution of operating points fills these regions more densely, but
does not make their union convex.

All links follow the single-click model 
\eqref{eq:single-click-tradeoff}--\eqref{eq:single-click-operating-point},
with $R_{\mathrm{rep}}=10^8~\mathrm{s}^{-1}$. The effective transmissions
on the first links of $P_1$, $P_2$, and $P_3$ are, respectively,
$(\eta_1,\eta_2,\eta_3)=10^{-3}(1,0.8,10/3)$; the second link of $P_i$
has transmission $4\eta_i$. Both links on each path are independently
tunable over $W\in[0.900,0.910]$, $[0.925,0.940]$, and $[0.988,0.992]$,
respectively. These path-dependent values are chosen for illustrative purposes so that the attainable regions of these paths do not overlap.
We sample $A=3,17,257$ uniformly spaced operating points per link, and use $W_{\min}=0.8$ and
$K=4463$ fidelity-loss bins. Utilities are evaluated using the recovered
e2e Werner parameter; exhaustive enumeration provides the sampled frontier
and checks the QLSP and adaptive-surrogate solutions.

For $U_{\mathrm{th}}$ with $W_{\mathrm{req}}=0.88$, QLSP selects an
unsupported configuration (teal stars in the figure) on $P_2$ at every sampled resolution, whereas
the surrogate selects $P_3$. At $A=257$, their rates are approximately
$6.05\times10^3$ and $4.80\times10^3$ pairs/s, respectively. Finer sampling
thus resolves the intermediate island but does not bring its sampled optimum
onto the lower convex envelope.

\begin{figure*}[t]
\centering
\includegraphics[width=\textwidth]{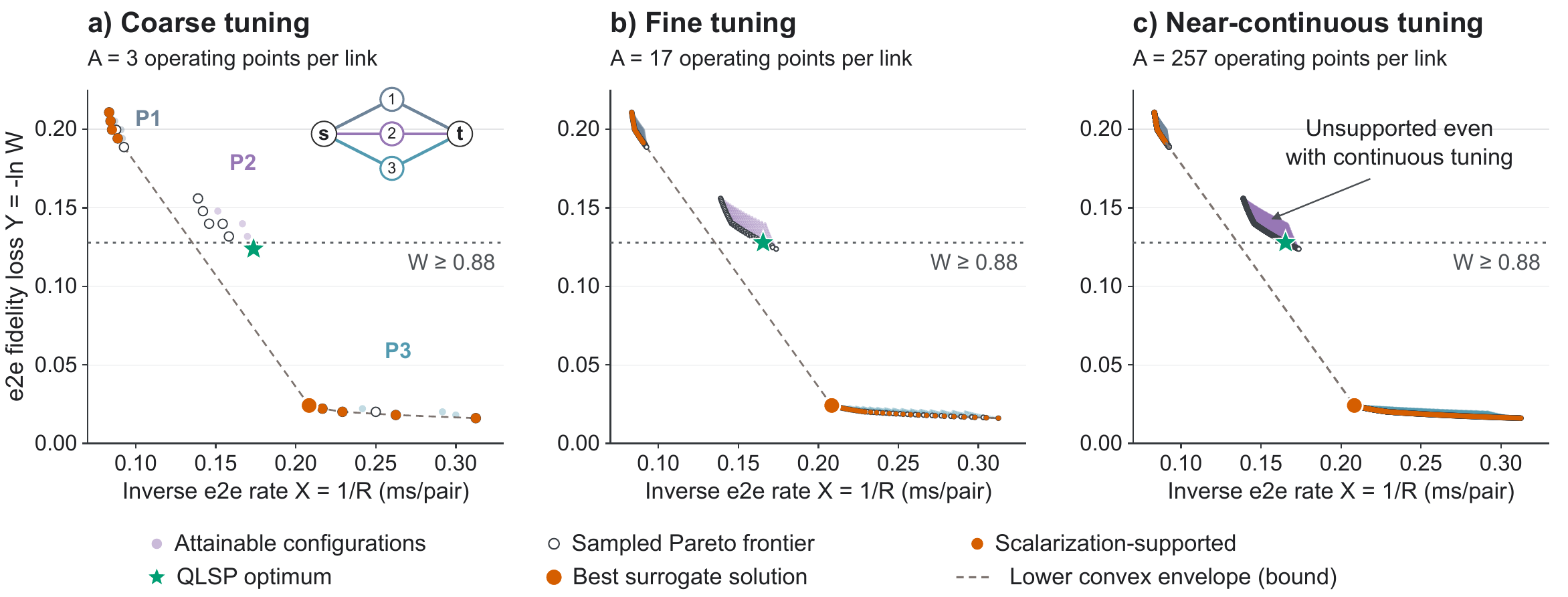}
\caption{\textbf{Refining operating points within fixed physical paths.}
The same three paths are sampled with a) $A=3$, b) $A=17$, and
c) $A=257$ operating points per link. Pale points show attainable
configurations; hollow circles mark sampled Pareto points and orange
circles mark scalarization-supported points. The horizontal line marks $W_{\mathrm{req}}=0.88$, with feasible
configurations below it. Stars identify QLSP solutions; larger orange
circles identify the best feasible surrogate solutions.}
\label{fig:continuous-scalarization-gap}
\end{figure*}

\section*{Efficient layered graphs approach for continuous operating points}

Continuous operating points $\mathcal{A}_e$ can be incorporated without
discretizing the operating-point curve when the per-link inverse-rate
cost has suitable convex structure in $y=-\ln W$. In this case, each
edge relaxation is a convex min-plus convolution whose transition
matrix is Monge. A hop-indexed dynamic program can compute all layer
transitions through an edge in $O(K\log K)$ time using a simple
divide-and-conquer algorithm, or in $O(K)$ time using the ``SMAWK" algorithm~\cite{aggarwal1986geometric}. The resulting complexities
are $O(L|\mathcal E|K\log K)$ and $O(L|\mathcal E|K)$,
respectively, rather than the $O(L|\mathcal E|K^2)$ cost of explicitly
examining all pairs of layers.

Let $D_h(v,k)$ denote the minimum cumulative inverse rate
$X=1/R$ of any path from $s$ to $v$ that uses at most $h$
physical edges and terminates in fidelity-loss layer $k$.

For each edge $e\in\mathcal{E}$, define
\begin{equation}
    c_e(j)
    =
    \min_{\substack{a\in\mathcal A_e:\\
    \left\lceil -\ln W_{e,a}/\Delta \right\rceil=j}}
    \frac{1}{R_{e,a}},
\end{equation}
with $c_e(j)=\infty$ if no operating point induces layer
increment $j$. Here $\mathcal A_e$ may be a continuous feasible
set of operating points. The dynamic program is initialized as
\begin{equation}
    D_0(s,0)=0,
    \qquad
    D_0(v,k)=\infty
    \quad\text{for }(v,k)\neq(s,0).
\end{equation}

For $h=1,\ldots,L$, and for $k=0,\dots,K$, define the relaxation through an edge
$e=(u,v)$ as
\begin{equation}
\label{eq:edge-relaxation}
    T_{h,e}(k)
    =
    \min_{0\leq j\leq k}
    \left\{
        D_{h-1}(u,k-j)+c_e(j)
    \right\}.
\end{equation}
The equation above is a min-plus convolution.
The dynamic-programming recurrence is then
\begin{equation}
\label{eq:hop-indexed-dp}
    D_h(v,k)
    =
    \min\left\{
        D_{h-1}(v,k),\;
        \min_{e=(u,v)\in\mathcal E} T_{h,e}(k)
    \right\}.
\end{equation}

After $L$ iterations, each reachable destination layer $k$ yields a
candidate path with rate
\begin{equation}
    R_k=\frac{1}{D_L(t,k)}.
\end{equation}
Its layer represents fidelity loss approximately $k\Delta$, or
$W_k\approx e^{-k\Delta}$. As in the standard QLSP algorithm, the
exact fidelity of the recovered path can instead be computed from its
selected operating points and used to evaluate the final utility.

When $c_e(j)$ is convex in $j$, the edge relaxation
\eqref{eq:edge-relaxation} is a Monge min-plus convolution, allowing
all destination-layer values to be computed jointly rather than
examining all $O(K^2)$ source--destination layer pairs.

\paragraph{Convexity of the per-link cost}
The Monge acceleration applies whenever the minimum inverse-rate cost
$c_e(j)$ is convex in the fidelity-loss increment $j$. This condition
holds, in particular, for the continuous single-click operating-point
model in Eqs.~(2)--(4). Recall that
\[
    F(\alpha)=1-\alpha,
    \qquad
    R(\alpha)=2R_{\rm rep}\eta_e\alpha,
    \qquad
    W(\alpha)=1-\frac{4}{3}\alpha,
\]
where $\eta_e=\eta(L_e)$ is the transmissivity of edge $e$.
Eliminating $\alpha$ gives the affine rate--Werner tradeoff
\begin{equation}
    R_e(W)
    =
    \frac{3}{2}R_{\rm rep}\eta_e(1-W).
\end{equation}
Writing $y=-\ln W$, so that $W=e^{-y}$, the inverse-rate cost becomes
\begin{equation}
\label{eq:single-click-convex-cost}
    x_e(y)
    =
    \frac{1}{R_e(y)}
    =
    \frac{2}
    {3R_{\rm rep}\eta_e(1-e^{-y})}.
\end{equation}
This function is decreasing and strictly convex (in the relevant region of $y$, which is nonnegative for $W\in[0,1]$):
\begin{align}
    x'_e(y)
    &=
    -\frac{2e^{-y}}
    {3R_{\rm rep}\eta_e(1-e^{-y})^2}
    <0,\\
    x''_e(y)
    &=
    \frac{2e^{-y}(1+e^{-y})}
    {3R_{\rm rep}\eta_e(1-e^{-y})^3}
    >0.
\end{align}
Since $x_e(y)$ is decreasing, the minimum inverse-rate cost within
fidelity-loss bin $j$ is attained at the largest feasible value of
$y$ in that bin. For bins whose upper endpoint $j\Delta$ lies within
the feasible operating range,
\begin{equation}
    c_e(j)=x_e(j\Delta).
\end{equation}
Sampling a convex function on a uniform grid preserves discrete
convexity, and hence
\begin{equation}
\label{eq:discrete-convexity}
    2c_e(j)
    \leq
    c_e(j-1)+c_e(j+1).
\end{equation}

\paragraph{Monge structure}
For a fixed edge $e=(u,v)$ and hop iteration $h$, define the
transition matrix over feasible source--destination layer pairs as
\begin{equation}
    M_{h,e}(k,i)
    =
    D_{h-1}(u,i)+c_e(k-i).
\end{equation}
The edge relaxation is its row minimum:
\begin{equation}
    T_{h,e}(k)=\min_i M_{h,e}(k,i),
\end{equation}
which is obtained by re-parameterizing $j=k-i$ in \eqref{eq:edge-relaxation}.
Then, discrete convexity in
\eqref{eq:discrete-convexity} gives
\begin{align}
    &M_{h,e}(k,i)+M_{h,e}(k+1,i+1) \nonumber\\
    &\qquad\leq
    M_{h,e}(k,i+1)+M_{h,e}(k+1,i),
\end{align}
so the transition matrix has Monge structure over its feasible
entries. Consequently, if $i^\star(k)$ denotes the smallest
source-layer index attaining the minimum for destination layer $k$,
then
\begin{equation}
\label{eq:monotone-minimizers}
    i^\star(k)\leq i^\star(k+1).
\end{equation}
Thus, the optimal source-layer index moves monotonically forward as
the destination layer increases.

\paragraph{An $O(K\log K)$ divide-and-conquer algorithm}
The monotonicity property \eqref{eq:monotone-minimizers} yields a
simple divide-and-conquer algorithm for computing all row minima.
Suppose we wish to compute $T_{h,e}(k)$ for destination layers
$k\in[k_\ell,k_r]$ and know that their minimizing source layers lie
in $i\in[i_\ell,i_r]$. We first evaluate the middle layer
\[
    k_m=\left\lfloor\frac{k_\ell+k_r}{2}\right\rfloor
\]
by scanning the feasible source indices in $[i_\ell,i_r]$ and
obtaining
\[
    i_m=i^\star(k_m).
\]
Monotonicity implies that all destination layers to the left of
$k_m$ have minimizers in $[i_\ell,i_m]$, while all layers to the
right have minimizers in $[i_m,i_r]$. We therefore recurse on
\[
    [k_\ell,k_m-1]\times[i_\ell,i_m]
\]
and
\[
    [k_m+1,k_r]\times[i_m,i_r].
\]
At each recursion depth, the total number of candidate source indices
examined is $O(K)$, and there are $O(\log K)$ depths. Hence all
transitions through one physical edge can be computed in
$O(K\log K)$ time using this simple implementation, giving an overall
hop-indexed complexity of
\begin{equation}
    O\!\left(L|\mathcal E|K\log K\right).
\end{equation}
The SMAWK linear-time Monge matrix-search algorithm~\cite{aggarwal1986geometric} further
reduces each edge relaxation to $O(K)$ and the overall complexity to
$O(L|\mathcal E|K)$.

\section{Additional Experiments}
\label{app:experiments}
\end{document}